\documentclass[envcountsame,envcountsect]{svmult}

\usepackage{amsmath,amssymb,calc,color}
\usepackage{graphicx}
\usepackage{epsfig,psfrag}
\usepackage{subfig}
\usepackage{mathptmx,helvet,makeidx,multicol,footmisc}
\usepackage{dotlessj}

\let\a=\alpha
\let\b=\beta

\let\s=\sigma
\let\t=\tau

\let\S=\Sigma

\newcommand{\del}{\partial}

\newcommand{\bea}{\begin{eqnarray}}
\newcommand{\eea}{\end{eqnarray}}
\newcommand{\bean}{\begin{eqnarray*}}
\newcommand{\eean}{\end{eqnarray*}}
\newcommand{\be}{\begin{equation}}
\newcommand{\ee}{\end{equation}}
\newcommand{\ben}{\begin{equation*}}
\newcommand{\een}{\end{equation*}}
\def\ba#1\ea{\begin{align}#1\end{align}}
\def\ban#1\ean{\begin{align*}#1\end{align*}}
\newcommand{\bma}{\begin{pmatrix}}
\newcommand{\ema}{\end{pmatrix}}

\newcommand{\hlf}{\frac{1}{2}}

\newcommand{\cD}{\mathcal{D}}

\newcommand{\cH}{\mathcal{H}}

\newcommand{\cM}{\mathcal{M}}

\newcommand{\cO}{\mathcal{O}}

\newcommand{\cR}{\mathcal{R}}

\newcommand{\cU}{\mathcal{U}}

\newcommand{\IH}{\mathbb{H}}

\newcommand{\La}{\Lambda}
\newcommand{\la}{\lambda}

\newcommand{\G}{\Gamma}
\newcommand{\e}{\epsilon}
\newcommand{\ve}{\varepsilon}

\newcommand{\dd}{\delta}

\newcommand{\m}{\mu}
\newcommand{\n}{\nu}

\newcommand{\om}{\omega}
\renewcommand{\t}{\theta}

\newcommand{\C}[1]{$(\ref{#1})$}

\def\ie{{i.e.}}

\def\eg{{e.g.}}

\newlength{\bredde}
\def\slash#1{\settowidth{\bredde}{$#1$}\ifmmode\,\raisebox{.15ex}{/}
\hspace*{-\bredde} #1\else$\,\raisebox{.15ex}{/}\hspace*{-\bredde}
#1$\fi}

\def\Im{{\rm Im ~}}
\def\Re{{\rm Re ~}}
\newcommand{\bra}[1]{\langle{#1}|}
\newcommand{\ket}[1]{|{#1}\rangle}
\newcommand{\braket}[2]{\langle{#1}|{#2}\rangle}
\newcommand{\ena}{\end{eqnarray}}
\newcommand{\beqa}{\begin{eqnarray}}
\newcommand{\eeqa}{\end{eqnarray}}

\def\G{\Gamma}

\def\bes #1\ees{\begin{split}#1\end{split}}

\renewcommand{\b}{\beta}
\newcommand{\g}{\gamma}

\newcommand{\ibar}{{\bar \imath}}
\newcommand{\jbar}{\jmath\hskip-0.25mm \bar{}\hskip.2mm}

\newfont{\goth}{ygoth.tfm scaled 1200}                   

\def\a{\alpha}
\def\b{\beta}
\def\e{\epsilon}
\def\th{\theta}

\def\g{\gamma}

\def\i{\iota}
\def\j{\psi}

\def\m{\mu}
\def\n{\nu}

\def\s{\sigma}
\def\t{\tau}

\def\G{\Gamma}

\def\S{\Sigma}
 \numberwithin{equation}{section}

\def\1{{(1)}}
\def\2{{(2)}}
\def\3{{(3)}}

\def\1{{\bf 1}}
\def\a{{\alpha}}

\def\R{{\mathbb R}}
\def\Z{{\mathbb Z}}
\def\CC{{\mathbb C}}
\renewcommand{\P}{{\mathbb P}}

\def\1{{\bf 1}}
\def\3{{\bf 3}}
\def\7{{\bf 7}}
\def\2{{\bf 2}}
\def\8{{\bf 8}}

\renewcommand{\thesection}{\arabic{section}}

\begin{document}
\title*{Mirror Symmetry in Physics: The Basics}
\author{Callum Quigley}
\institute{Callum Quigley \at Department of Mathematical and Statistical Sciences,\\  University of Alberta, Edmonton, AB T6G 2G1, Canada, \email{cquigley@ualberta.ca}}

\maketitle

\abstract*{These notes are aimed at mathematicians working on topics related to mirror symmetry, but are unfamiliar with the physical origins of this subject.
We explain the physical concepts that enable this surprising duality to exist, using the torus as an illustrative example.
Then, we develop the basic foundations of conformal field theory so that we can explain how mirror symmetry was first discovered in that context.
Along the way we will uncover a deep connection between conformal field theories with (2,2) supersymmetry and Calabi-Yau manifolds.
(Based on lectures given during the \emph{Thematic Program on Calabi-Yau Varieties: Arithmetic, Geometry and Physics} at the Fields Institute in Toronto,
October 10--11, 2013.)}

\abstract{These notes are aimed at mathematicians working on topics related to mirror symmetry, but are unfamiliar with the physical origins of this subject.
We explain the physical concepts that enable this surprising duality to exist, using the torus as an illustrative example.
Then, we develop the basic foundations of conformal field theory so that we can explain how mirror symmetry was first discovered in that context.
Along the way we will uncover a deep connection between conformal field theories with (2,2) supersymmetry and Calabi-Yau manifolds.
(Based on lectures given during the \emph{Thematic Program on Calabi-Yau Varieties: Arithmetic, Geometry and Physics} at the Fields Institute in Toronto,
October 10--11, 2013.)}


\section{Introduction}\label{sec:intro}

String theory lies right at the interface of physics and mathematics. Researchers on both sides have consistently benefited from close interaction with one another, with advances in one field unlocking deep secrets in the other.
Perhaps no other topic exemplifies this fruitfulness like mirror symmetry does.

In physical terms, mirror symmetry is an example of a duality, meaning an exact equivalence between two seemingly different physical systems. The advantage of dual systems is that often when one is difficult to compute with, the other is simple, which makes dualities an extremely powerful tool.
For our purposes, we can think of mirror symmetry as an equivalence between the physics of string theory on two different Calabi-Yau manifolds.
If $X$ and $Y$ are Calabi-Yau manifolds related by mirror symmetry, we will call $(X,Y)$ a {\it mirror pair}. Roughly speaking, $(X,Y)$ form a mirror pair if the (complexified) K\"ahler structure of $X$ is equivalent to the complex structure of $Y$, and vice-versa.
Since Calabi-Yau manifolds typically come in large families, with many tunable parameters ({\it moduli}), what we really mean is that mirror symmetry is a equivalence of families of Calabi-Yaus and the corresponding moduli spaces match.
One immediate consequence is that if $(X,Y)$ form mirror pair of Calabi-Yau $n$-folds, then their Hodge numbers are related: $h^{1,1}(X) = h^{n-1,1}(Y)$ and $h^{n-1,1}(X) = h^{1,1}(Y)$. In particular, $X$ and $Y$ (typically) are topologically distinct.
The fact that string theory on two manifolds with different topologies can give the exact same results is both remarkable, and surprising.

Calabi-Yau manifolds first gathered widespread attention in the physics community after~\cite{Candelas:1985en}, where it was shown that string theory ``compactified" on Calabi-Yau three-folds can give rise to realistic models of our world.\footnote{The brief history of mirror symmetry that we are about to cover is far from complete, and the references we provide are far from exhaustive.}
Such string theories are described by so-called $N=(2,2)$ superconformal field theories in two dimensions, which will be discussed in great detail in these notes.
Mirror symmetry was first noticed in the study of certain $(2,2)$ models~\cite{Dixon:1987bg}, where it was noted that the underlying geometry could not be uniquely determined from the data of the field theory.
Instead, there was an ambiguity in $h^{1,1}\leftrightarrow h^{2,1}$.
Soon afterwards, this was found confirmed in a much larger class of $(2,2)$ models~\cite{Lerche:1989uy}, and it was conjectured to be a generic feature (subject to certain natural conditions).
In~\cite{Candelas:1989hd}, thousands of Calabi-Yau three-folds were constructed and analyzed, and most were found to come in pairs that differed under the interchange  $h^{1,1}\leftrightarrow h^{2,1}$,
and methods for constructing mirror pairs were developed~\cite{Greene:1990ud,Batyrev:1994hm,Batyrev:1994pg}.
This gave tremendous evidence for the mirror conjecture, though a full proof was still far off. A physics based ``proof" of mirror symmetry would not appear until~\cite{Hori:2000kt}.

Mirror symmetry entered the mathematics community when the authors of~\cite{Candelas:1990rm}\ successfully used mirror symmetry to predict the numbers of rational curves in certain Calabi-Yau manifolds, thereby solving long-standing problems in enumerative geometry.\footnote{As the story goes, these results were presented at a conference at the MSRI in Berkeley where it was pointed out that one of their numbers was in disagreement with recent computations of mathematicians (using traditional and more rigorous techniques). However, an error in the computer code of the mathematicians was soon discovered, and the predictions on mirror symmetry were verified~\cite{Shape}.}
Most mathematicians working on this topic confine themselves to a narrower, though sharper, definition of mirror symmetry, due largely to the work of Witten.
In~\cite{Witten:1989ig}, extending his earlier ideas of~\cite{Witten:1988xj}, Witten introduced the notion of topological string theory: a simplified version of the full theory.
In~\cite{Witten:1991zz}\ and~\cite{Vafa:1991uz}, it was shown that a sharper version of the mirror conjecture exists in the topological theories, and this is usually taken as the mathematical definition of mirror symmetry.
Building on work in~\cite{Kontsevich:1994na}, this refined (mathematical) mirror conjecture was proven in~\cite{Givental96,1997alg.geom..1016G}, and then in greater detail in~\cite{1997alg.geom.12011L,1999math......5006L,1999math.....12038L,2000math......7104L}.

Aside from possible passing remarks, we will not discuss the topological string formulation, open string mirror symmetry~\cite{Witten:1992fb}, homological mirror symmetry~\cite{Kontsevich:1994dn}, the SYZ conjecture~\cite{syz}, heterotic mirror symmetry~\cite{Melnikov:2012hk}, or any of the other more modern directions this subject has developed.
These notes will focus exclusively on the original (more physical) formulation of the mirror conjecture from~\cite{Lerche:1989uy}.
However, see~\cite{Clader:2014kfa}\ for a recent review of several of the more mathematical formulations of mirror symmetry.
The rest of these notes are organized as follows.
In Sect.~\ref{sec:torus}, using the simple example of a torus, we explain how string theory modifies our conventional notions of geometry so that an equivalence as absurd sounding as mirror symmetry could ever possibly be true.
This worked example is presented rather informally (\ie\ in the ``physics style"), as are other examples in subsequent sections.
In Sect.~\ref{sec:conf}, we take a detour to introduce the basic fundamentals of conformal field theory, and in an attempt to appeal to the intended mathematical audience, we have tried our best as to present the material in a familiar manner (definitions, propositions, theorems, proofs, etc.).
We hope the reader will forgive any of this simple physicist's gross misuse of these structures. Readers already familiar with these concepts may wish to skip, or skim through, this second section. Finally, in Sect.~\ref{sec:SCFT}\ we study supersymmetric conformal field theories, which is the setting where mirror symmetry first arose.
We will explain how the algebraic structures of these field theories share many properties with Calabi-Yau manifolds, and how the profound implications of mirror symmetry arise from a simple ambiguity in the physics.

\begin{acknowledgement}
I would like to thank all the organizers of this program for the invitation to give these lectures, especially Noriko Yui for her immense patience as I prepared these notes well past the deadlines.
Also, I would like to thank Vincent Bouchard, Peter Overholser, and Johannes Walcher for helpful discussions, and Thomas Creutzig for many discussions and feedback on earlier drafts.
This work is supported by a PIMS postdoctoral fellowship and an NSERC Discovery Grant.
\end{acknowledgement}

\section{Warm-up: The Torus}\label{sec:torus}
To understand how string theory on very different looking backgrounds can give rise to identical physics, it is best to start with an simple example.
Although the most interesting examples of mirror symmetry occur for Calabi-Yau manifolds of complex dimension three (and higher, though $K3$ surfaces have their own interesting features),
it turns out that the simplest possible Calabi-Yau, the complex torus or elliptic curve, already demonstrates the essential features with minimal extra complications.
In this section we will present a rather informal (by most mathematicians' standards) discussion of mirror symmetry for the torus, emphasizing some of the physical features of string theory on this background that make mirror symmetry possible.
We begin in Sects.~\ref{ss:first}\ and~\ref{ss:complex}\ with a review of the classical moduli spaces of complex tori. We will see that mirror symmetry is simply not feasible in this restricted setting.
However, by introducing novel concepts from string theory, in Sects.~\ref{ss:Kahler}\ and~\ref{ss:Tdual}, we will see how such an equivalence can arise.
We will summarize our heuristic discussion of mirror symmetry in Sect.~\ref{ss:summary}, with an aim to generalize these results to higher dimensional Calabi-Yau manifolds.
The material we are about to present is quite standard, and can be found in most modern string theory textbooks (such as~\cite{Polchinski:1998rq,Becker:2007zj}).

\subsection{A First Look at the Moduli}\label{ss:first}
Let's start with the very basics. When we first learn about tori, we are told to think of the surface of a donut.
This certainly makes the topology of $T^2$ evident, but as an object embedded in $\R^3$ it obscures the fact that $T^2$ is actually flat.
When we are a little older and wiser, we are told instead to think of the torus as a square with opposite sides identified, as in Fig.~\ref{fig:T2}.
\begin{figure}[h]
\centering
\includegraphics[trim= 0 500 0 75 ,clip,scale=0.5]{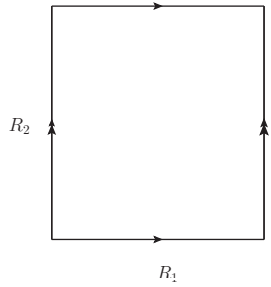}
\caption{\textit{A torus with sides of length $R_1$ and $R_2$}}
\label{fig:T2}
\end{figure}
This makes the flatness evident, as well as its product structure: $T^2\simeq S^1_{R_1}\times S^1_{R_2}$, where $S^1_{R}$ is a circle of radius $R$. A rather obvious, but very important, fact is that there is a (flat) torus for every choice of $R_{1,2}>0$.
We say that $R_1$ and $R_2$ are the {\it moduli} of the torus, meaning they are the parameters that we need to very in order to sweep out the entire family of possible tori. We call this space of possible tori the {\it moduli space} of the torus, which we will denote by $\cM(T^2)$. Based on our discussion so far, you might conclude that $\cM(T^2) = \R_+\times\R_+$, but this is not quite right as we will now explain.

Instead of $R_1$ and $R_2$, it is helpful to parameterize the torus in terms of the product and ratio of these radii. We define
$$
A = R_1R_2,\qquad \t_2 = R_2/R_1,
$$
where $A$ gives the total area of the torus, and $\t_2$ controls its shape. (The origin of the nomenclature $\t_2$ will be clear in a moment.) In Fig.~\ref{fig:deforms}\ we depict independent deformations of $A$ and $\t_2$ in two cases: $(a)~(R_1,R_2)\rightarrow(2R_1,2R_2)$ and $(b)~(R_1,R_2)\rightarrow(\hlf R_1,2R_2)$.
It is a general feature of Calabi-Yau manifolds that the moduli can be organized into two classes of these types: namely, {\it K\"ahler deformations} that control sizes, and {\it complex structure deformations} that control shapes. The total moduli space of any Calabi-Yau $X$ factorizes into the product of moduli spaces for K\"ahler and complex structure deformations:
$$
\cM(X) = \cM_K(X)\times\cM_{cs}(X).
$$
\begin{figure}[h]
\centering
\subfloat[][]{
\includegraphics[trim= 0 450 50 50 ,clip,scale=0.27]{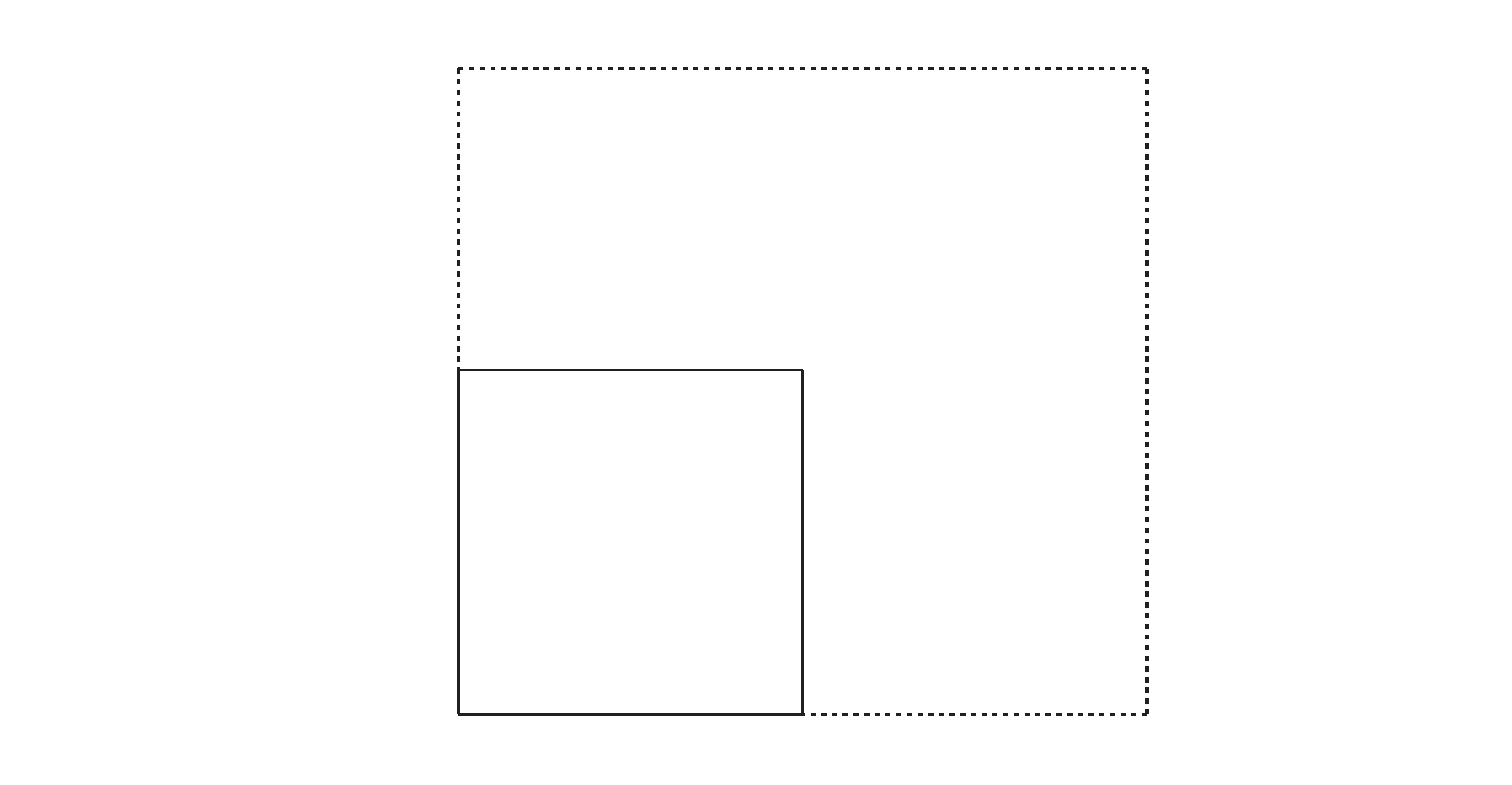}
\label{fig:Size}}
\
\subfloat[][]{
\includegraphics[trim= 50 450 50 0 ,clip,scale=0.27]{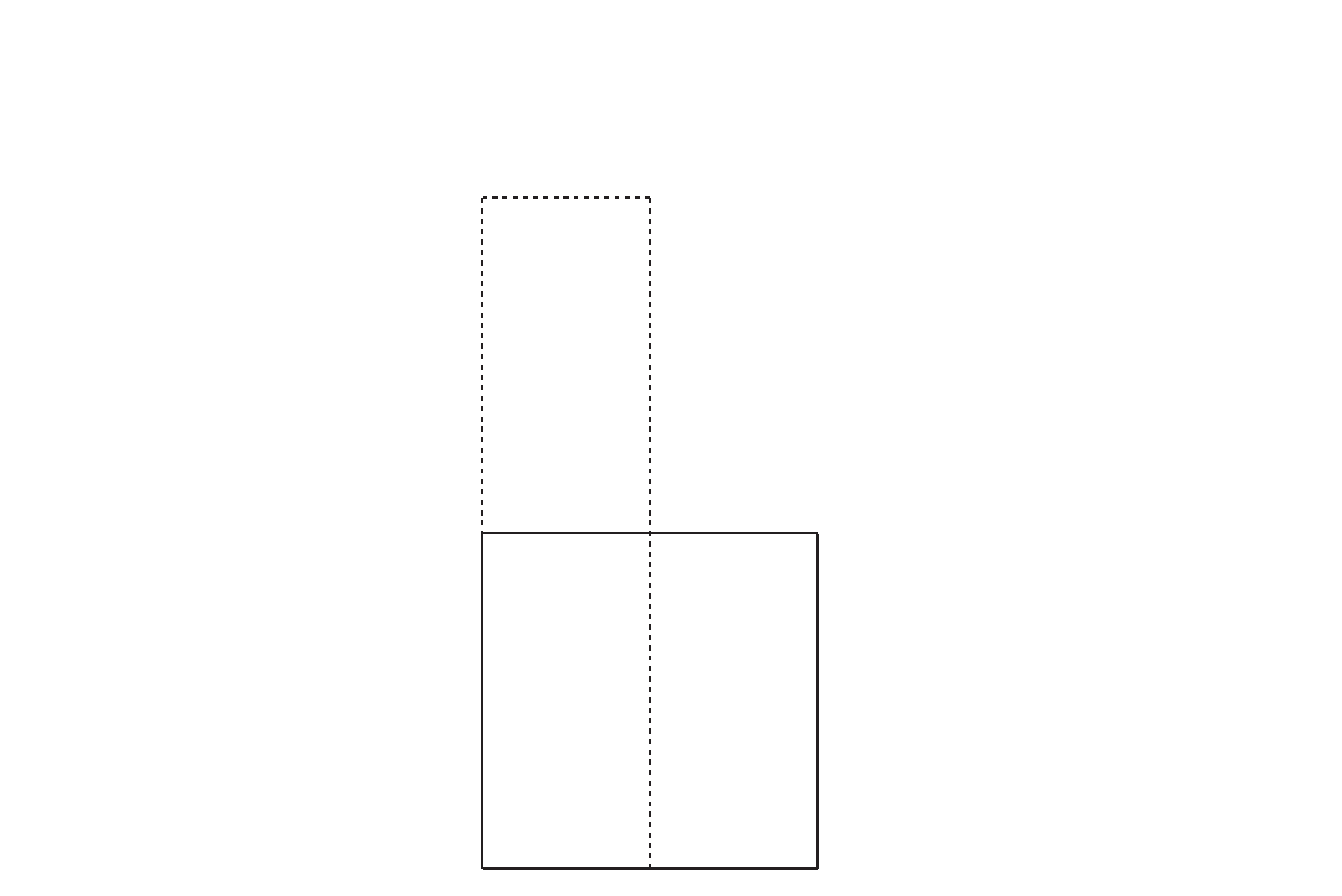}
\label{fig:shape}}
\caption{\textit{(a) K\"ahler deformations change the size of the torus, but leave its shape invariant\
(b) Complex structure deformations preserve the area, but change the overall shape}}
\label{fig:deforms}
\end{figure}
This is something we have already seen for the torus, where we argued that each factor is just $\R_+$. The claim of mirror symmetry is that these two factors are interchanged under the duality map. That is, if $(X,Y)$ form a mirror pair, then
$$
\cM_K(X) = \cM_{cs}(Y),\quad {\rm and}\quad \cM_{cs}(X)=\cM_K(Y).
$$
Note that this leaves the total moduli space invariant, \ie~ $\cM(X)=\cM(Y)$.
This certainly makes sense for the torus, since it is the unique Calabi-Yau in one (complex) dimension. Thus, if $X=T^2$ then its mirror must be another torus, $Y=\check T^2$, and the mirror map simply interchanges the factors of $\R_+$ in the moduli space. Everything works out fine! End of story, right?

\subsection{The Complex Structure Modulus}\label{ss:complex}
Of course, we have been far too glib in our discussion so far. As we all know, a torus carries a natural complex structure whose deformations are parameterized by a complex quantity, not just the real modulus $\t_2$ that we have discussed so far.
Indeed, a convenient way to realize this structure is to construct the torus as a quotient $\CC/\La$, where $\La \simeq \Z\oplus\t\Z$ is a rank two lattice and $\t=\t_1+\I\t_2$ is a complex parameter.
So more generally, instead of Fig.~\ref{fig:T2}\ we should think of tori as parallelograms with opposite sides identified, as in Fig.~\ref{fig:lattice}.
However, $\cM_{cs}(T^2)$ is not given by all $\t\in\CC$. First of all, we can restrict to the upper half plane:
$$
\mathbb{H} :=\{\t\in\CC{\big|}\Im\t>0\},
$$
since complex conjugation of $\t$ produces isomorphic tori (and if $\t$ is real, then the torus degenerates), but we can restrict $\t$ even further.
Since we have identified all points $z\in\CC$ under $z\simeq z+m +n\tau$, for all integers $m$ and $n$, then it is clear that $\t$ and $\t+1$ produce the same lattice.
\begin{figure}[h]
\centering
\includegraphics[trim= 0 500 0 50 ,clip,scale=0.5]{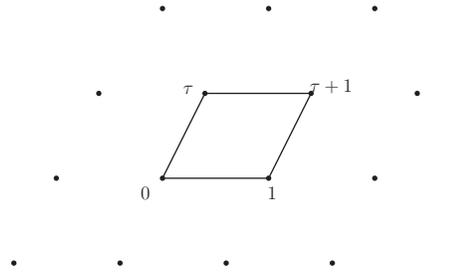}
\caption{\textit{A torus in the complex plane, constructed as $\CC/(\Z\oplus\t\Z)$}}
\label{fig:lattice}
\end{figure}
Similarly, up to an overall scaling\footnote{Note that rescaling $z$ is equivalent to scaling the total area, $A$, and therefore is  a K\"ahler (not a complex structure) deformation.} of $z$, $\t$ and ${-1/\t}$ also define equivalent tori. Together, these actions
$$
T:~\t\mapsto \t+1,\quad{\rm and}\quad S:~\t\mapsto-\frac{1}{\t},
$$
generate the {\it modular group} $\G = SL(2,\Z)$, which in general acts as
$$
\G: \t\mapsto \frac{a\t+b}{c\t+d}
$$
for all integers $a,b,c,d$ such that $ad-bc=1$. Thus, the orbit under $\G$ of any given $\t\in\IH$  consists of an infinite number of points, all corresponding to tori with identical complex structures.
To label the set of inequivalent complex structures we should consider the (right) coset $\G\backslash\IH$, which we may take to be the region
$$
F_0 = \left\{ \t\in\IH\left|~|\t|\geq1,~  -\tfrac{1}{2}<\Re\t\leq\tfrac{1}{2}  \right. \right\}.
$$
The space $F_0$ is often called the {\it fundamental domain of $\G$}, and it is but one representative of the coset $\G\backslash\IH$. Any $\t\in\IH$ can be mapped (via $\G$) to a unique $\t_0\in F_0$, and similarly $\IH$ can be tiled by the (infinite number of) images of $F_0$ under $\G$.
Thus we learn that the full moduli space of complex structures for the torus is actually
$$
\cM_{cs}(T^2) = \G\backslash\IH \simeq F_0.
$$
However, now we see a problem. The moduli space of K\"ahler deformations for the torus is just $\cM_K(T^2) = \R_+$, corresponding to the total area, but this is not even the same dimension as $\cM_{cs}(T^2)$. How can these two spaces be interchanged under mirror symmetry when they are so different?

\subsection{The Complexified K\"ahler Modulus}\label{ss:Kahler}
In fact, as a statement about classical geometry the claims of mirror symmetry are blatantly false: the moduli spaces $\cM_{cs}(T^2)=\G\backslash\IH$ and $\cM_K(T^2)=\R_+$ are simply not of the same dimension, and there is no symmetry which could ever interchange them. However, our interest lies not in the standard geometry of points that has been studied for millennia, but rather in its `stringy' generalization (which is a subject only decades old).
We will now explain how a given geometry can appear very differently when probed by a string as opposed to a point particle. It is only in this generalized setting, where standard geometric notions begin to break down, that mirror symmetry even begins to makes sense.

Earlier we saw that there is a natural complexification of the modulus $\t_2=R_2/R_1$, as in Fig.~\ref{fig:T2}, which comes from including the effects of tilting the torus, as in Fig.~\ref{fig:lattice}. What we need now is an analogous complexification of the area $A$, but what could that mean? To make sense of this, it helps to first make a more slightly abstract definition of the K\"ahler moduli. We write
$$
A = \int_{T^2} \omega,
$$
where $\omega\in H^2(T^2)$ is the so-called K\"ahler form, which is related to the metric $g$ by the (almost) complex structure $J$: $g(x,y) = \omega(x,Jy)$. In local holomorphic coordinates, which physicists love to use so much but mathematicians are not fond of, we can write $\omega = \I g_{z\bar z} \D z \D \bar z$ with $g_{z\bar z}=A\in\R_+$.
So what we are looking for is a natural partner to the Hermitian two-form $\omega$, which effectively turns $g_{z\bar z}$ into a complex parameter. Fortunately, string theory provides exactly such an object, known as the $B$-field, and it naturally generalizes the point particle's gauge field.

Recall that the motion of a point-particle can be described by a set of embedding functions,
$$
x^\m(\xi):\g\hookrightarrow X,
$$
which map a curve $\g$ into the space $X$. Here $\xi$ is a coordinate that parameterizes the curve $\g$. The classical trajectory of a particle is found by extremizing the action
$$
S_0[x] = \int_\g \D\xi \sqrt{g(\del_\xi x,\del_\xi x)},
$$
whose solutions are geodesics on $X$. If we wish to couple this particle to some (abelian) vector bundle $V$ over $X$, we must introduce a connection one-form $A$, where we identify $A\sim A'= A+d\la$ since they lead to equivalent curvatures on $V$: $F=dA =dA'$. Now the motion of this charged particle is given by
$$
S[x] = S_0[x] +\int_\g \D\xi\ \del_\xi x^\mu A_\m(x) = S_0[x] +\int_\g x^*(A),
$$
where we see that the interaction term, between the particle and $A$, is nothing more than the pullback of $A$, from $X$ to $\g$, by the maps $x^\m$.

Now we generalize the previous discussion to the case of strings. A string propagating in time through $X$ sweeps out a two dimensional surface $\S$, so we must have embedding functions
$$
x^\m(\xi,\zeta): \S\hookrightarrow X.
$$
Let $S_0'[x]$ be the analog of $S_0[x]$ that extremizes the surface area of $\S$ in $X$. Instead of a one-form connection, the natural object to couple to the string is a two-form, $B$, by pulling it back to $\S$:
\be
S[x] = S'_0[x] + \int_\S x^*(B).\label{eqn:Baction}
\ee
Like the gauge connection $A$, $B$ also possess a gauge invariance, $B\sim B'= B+d\la$, where now $\la$ is a one-form, because this leads to equivalent curvatures $H:=dB = dB'$.
In fact, for our purposes we can impose that $B$ must also be closed, and therefore, just like $\omega$, $B \in H^{2}(X)$. The reason is that if $H=dB\neq0$, then $X$ could not be Calabi-Yau.\footnote{Physically, this follows from the fact that $H=dB\neq0$ would generate a finite energy density on $X$ which would preclude the possibility of a Ricci-flat solution to the Einstein equations. Mathematically, it can be shown that $H\neq0$ requires that $X$ is not K\"ahler, and so in particular will not be Calabi-Yau~\cite{strominger-torsion}.}

This gives us our natural partner for $\omega$, namely $B$, and we can form the so-called {\it complexified K\"ahler form}
$$
B+\I\omega = (b+\I A)\D z\D\bar z \in H^{1,1}(T^2,\CC),
$$
which we can integrate over the torus to obtain the complex modulus
$$
\varrho = \int_{T^2}(B+\I\omega) = b+\I A.
$$
This is good news, since now the two moduli spaces $\cM_K(T^2)$ and $\cM_{cs}(T^2)$ will at least be of the same dimension, namely complex dimension one. Furthermore, it makes perfect sense to restrict to $\varrho\in \IH$, since the area $A$ must be positive. However, in order for mirror symmetry to work, $\varrho$ should also be invariant under the modular group $\G=SL(2,\Z)$. Let's check if this is true.

An important point about the physics of the $B$-field, or more precisely its integrated value $b=\Re\varrho$, is that it need not be single-valued. The only thing that needs to be well-defined is the quantum mechanical path integral, which can formally be written as
\be\label{eqn:pathint}
Z_{T^2} = \int \left[{\cD}x\right]\ \E^{\I S[x]} = \int \left[{\cD}x\right]\ \exp(\I S'_0[x] + \I \int x^*(B)+\ldots) = \int \left[{\cD}x\right]\ \E^{2\pi \I b} \ldots\ ,
\ee
where $\left[{\cD}x\right]$ is a formal integration measure over the set of all embedding functions $x^\mu(\xi,\zeta)$, and we have used the action~\C{eqn:Baction}\ for the phase factor.
Mathematicians often cringe when the see expressions such as~\C{eqn:pathint}, since the measure $\left[{\cD}x\right]$ is not a well-defined quantity in any sensible manner.
Questions of well-posedness aside, the important point is the path integral (and therefore any relevant physical quantity) only depends on the $B$-field through the quantity $\exp(2\pi \I b)$.
In particular, integral shifts of $b$ leave the path integral (and therefore all physical quantities) invariant, and so
$$
T:\varrho \mapsto \varrho+1
$$
is a symmetry of the theory. More great news! All that remains to show now is that the $S$ transformations, which would send $\varrho$ to $-(1/\varrho)$, leaves the physics of the string invariant as well. Here we encounter a surprise: suppose we set $b=0$, then we have
\be
S:\, \varrho=\I A\, \mapsto\, -1/\varrho = \I/A.\label{eqn:S}
\ee
Invariance of $\varrho$ under $SL(2,\Z)$ requires that a string treats two tori with inversely related areas as being completely equivalent! This is in stark contrast to a point set description of the same geometries, where such an equivalence is purely nonsensical. Nevertheless, to a string those two spaces are indistinguishable, as we will now explain.

\subsection{T-duality}\label{ss:Tdual}
In fact, a variant of this inversion symmetry, $A\leftrightarrow 1/A$, already appeared (though not explicitly) back in Sect.~\ref{ss:first}. Recall that in our simplified description of the torus, we had only two real moduli: $A=R_1 R_2$ and $\t_2=R_2/R_1$, which each took values in $\R_+$.
Mirror symmetry exchanges these two moduli, with the consequence that the two tori with moduli $(A,\t_2)$ and $(\t_2,A)$ should lead to identical physics. At this point, the reader should notice the following peculiar fact:
$$
A  \leftrightarrow \t_2  \quad \Leftrightarrow \quad R_1 \leftrightarrow 1/R_1,
$$
which means that string theory on a circle of radius $R$ should be identical to string theory on a circle of radius $1/R$. This rather surprising fact has come to be known as {\it T-duality}. We do not have the time or space to fully derive T-duality as a symmetry of string theory here, in the rest of this section we will present one (fairly compelling) piece of evidence to support the notion. More details will also be presented in Sect.~\ref{ss:examples}.

Our goal will be to demonstrate that the (energy) spectrum of a (relativistic, quantum) string compactified on a circle of radius $R$ is the same as for a circle of radius $1/R$. While this fact alone does not constitute a proof, it will certainly lend credence to the claim that this is a symmetry of the full theory. Before considering the string, let us return once again to the (relativistic, quantum) point particle. The spectrum of such a particle, moving through the flat $d$ dimensional spacetime $\R^{1,d-1}$, is given by Einstein's famous relation
$$
E^2_{\rm part} = p^2 +m^2,
$$
where $p$ is the momentum of the particle, $m$ is the mass, and (like all good particle physicists) we have set the speed of light $c=1$.\footnote{Reinstating factors of $c$, this becomes $E^2 =c^2p^2 + m^2 c^4$, or at zero momentum simply $E=mc^2$.}
The mass is therefore any residual rest energy of the particle, be it an intrinsic mass or the result of its internal structure.\footnote{For example a proton is composed of three {\it quarks}, each carrying their own intrinsic mass, but together they are responsible a mere $1\%$ of a the proton's mass. The other $99\%$ arises from the internal binding energy (carried by {\it gluons}) that keeps the quarks from flying apart.}
A crucial feature of quantum mechanics is that when the particle propagates some distance, $\varDelta x$, it's wavefunction acquires a phase, $\E^{\I p\varDelta x}$, where (again, like every sensible physicist) we have set Plank's constant $\hbar=1$.
Now suppose we compactify on a circle of radius $R$. Every time the particle goes around this compact direction, it acquires the phase $\E^{2\pi \I pR}$. In order for these phases not to destructively interfere we must impose the quantization condition $p = n/R$, for $n\in\Z$. So, ignoring the (continuous) momentum in the remaining non-compact directions, we find that a particle's spectrum on $S^1_R$ is
\be
E^2_{\rm part} = \frac{n^2}{R^2} + m^2,\quad n\in\Z.\label{eqn:KK}
\ee
Notice that if $R$ were sufficiently small, so that an observer in the remaining $(d-1)$ dimensional spacetime could not see it, then the momentum $p=n/R$ would appear as a contribution to the rest mass energy, since it is not related to motion in the observed spacetime.

Now, what about the spectrum of a string? Because of its extended nature there is another form of energy carried by the string, which is not possible for a point particle, related to deforming its length.
Since the string has a finite tension $T$, which we can normalize to $T=1/2\pi$,\footnote{More precisely, in the natural units $\hbar=c=1$ the tension of the string is $T=(2\pi\a')^{-1}$, where $\a'$ has dimensions of area. This sets the fundamental length scale of a string, $\ell_s=\sqrt{\a'}$. If we keep $\a'$ explicit, as many physicists often do, then T-duality interchanges $R\leftrightarrow \a'/R$.}
then stretching the string by an amount $\varDelta L$ costs an energy $T\varDelta L = \varDelta L/2\pi$. Therefore, the spectrum of a string takes the form
$$
E_{\rm string}^2 = p^2 +\left(\frac{\varDelta L}{2\pi} \right)^2 +m^2.
$$
The strings we are interested in do not have any intrinsic mass, but they carry internal rest energy associated with their (quantized) vibrational modes.
A (massless) string can oscillate in any of the $d-2$ directions transverse to the surface $\S$ that it sweeps out in spacetime. In each direction there are an infinite number of vibrational modes (the Fourier modes), which we can label by $n=0,1,2,\ldots$, each of which contribute an energy $\sim n$. Finally, in the direction $\m=1,\ldots,d-2$, the $n$-th mode can have an arbitrary (but quantized) amplitude, $N_{\m n}$.
\begin{figure}[h]
\centering
\subfloat[][$(2,4)$]{
\includegraphics[trim= 150 475 200 50 ,clip,scale=0.40]{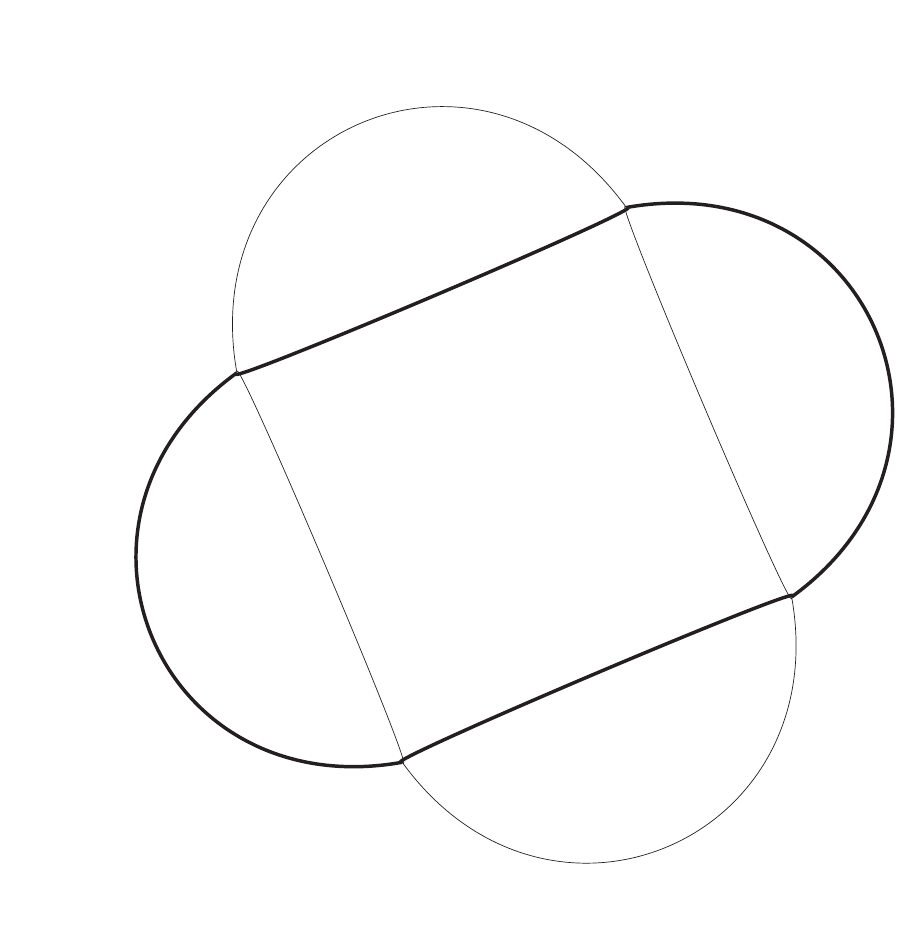}
\label{fig:4-2}}
\qquad
\subfloat[][$(4,2)$]{
\includegraphics[trim= 175 500 200 50 ,clip,scale=0.40]{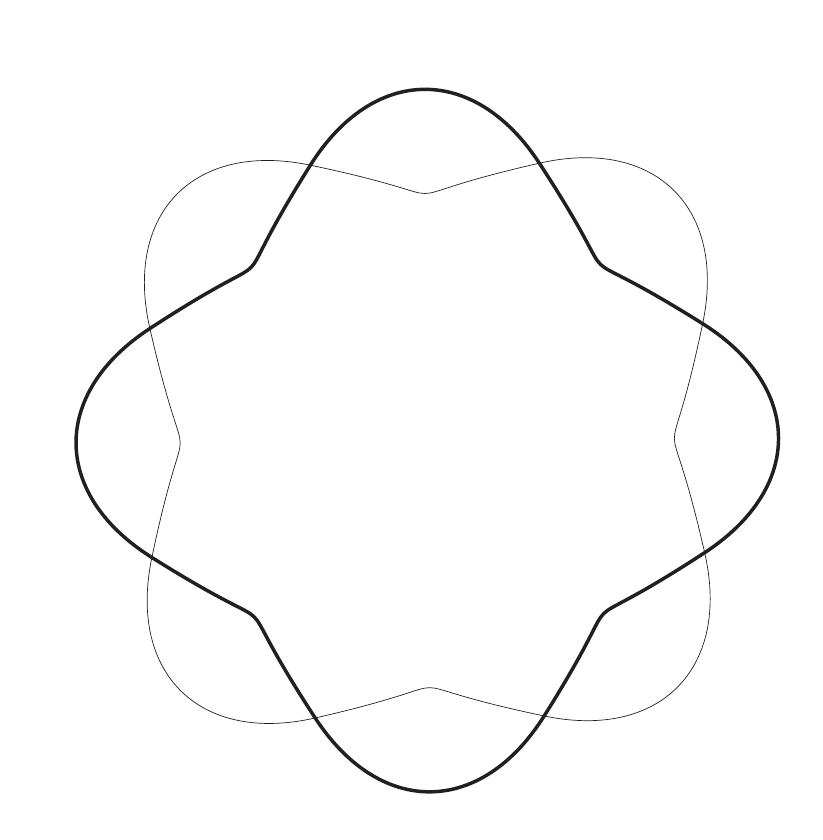}
\label{fig:2-4}}
\qquad
\subfloat[][$(8,1)$]{
\includegraphics[trim= 200 500 200 50 ,clip,scale=0.40]{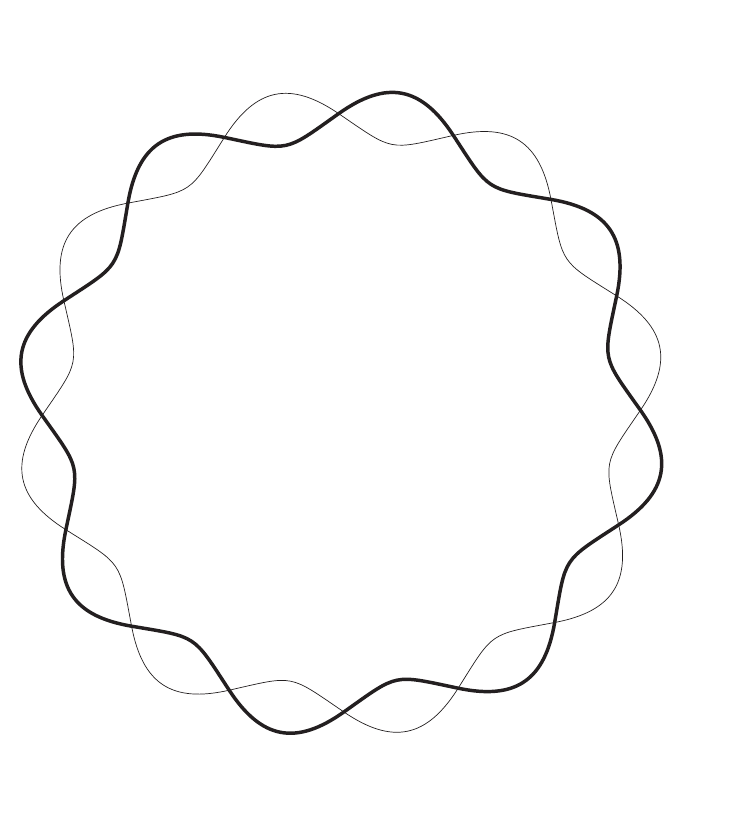}
\label{fig:1-8}}
\caption{\textit{Three vibrational patterns with numbers $(n,N_{\m n})$, all with $N=8$}}
\label{fig:vibs}
\end{figure}
In Fig.~\ref{fig:vibs}\ we have sketched three different vibrational patterns, corresponding to $(n,N_{\m n}) \in\{(2,4),(4,2),(8,1)\}$.
Altogether, the effective mass of the string is given by the total oscillation number, $N$, defined as
$$
m^2 = N := \sum_{\m=1}^{d-2} \sum_{n=1}^{\infty} n N_{\m n}.
$$
All three examples in Fig.~\ref{fig:vibs}\ have $N=8$. When we consider the string on a circle of radius $R$, once again the momentum along that direction is quantized: $p=n/R$. Since the string is extended, we can also consider configurations where the string  wraps (multiply) around the circle. The string can wrap any integral\footnote{If the wrapping number were not integral, then the string must start and end at different points, and would no longer be closed.} number of times, so $\varDelta L = 2\pi R w$, where (taking orientation into account) $w\in\Z$. Putting everything together, we see that the spectrum of a string on a $S^1_R$ is
$$
E_{\rm string}^2 = \frac{n^2}{R^2} + w^2 R^2 +N.
$$
Notice that the spectrum is invariant under $R\leftrightarrow 1/R$ if we simultaneously interchange $n\leftrightarrow w$, ie. we must swap momentum and winding numbers. This is the statement of T-duality.

The implications of T-duality for the structure of spacetime are very deep. Suppose that one day in the not-too-distant future a physics colleague runs up to you, excited with the news that the LHC has observed a tower of new particles with evenly spaced masses $m\sim n/R$.  In light of~\C{eqn:KK}, these are naturally interpreted as the set of momentum modes in some new compact dimension of space, which takes the shape of a circle of radius $R$. However, if string theory is a correct description of Nature, then this is not the only consistent interpretation of this exciting new data. Instead, these new states could correspond to the winding modes of a string on a circle of radius $1/R$. In fact, both interpretations would be equally valid and, furthermore, there would be no possible way to distinguish between them! After all, the only tool available to probe the size of this new dimension would be a string, which, as we have already seen, cannot differentiate between these two possibilities.
One way to interpret this fact is that there is no real meaning to circles with $R<1$ (ie. smaller than the size of the string) since they are always equivalent to circles with $R>1$. In this sense, at least for circular dimensions, string theory has a minimal length scale and there is no physical meaning to anything smaller.

\subsection{Summary}\label{ss:summary}
We have seen that $R\leftrightarrow1/R$ is a symmetry of circle compactifications in string theory. Returning now to the case of interest, $T^2$, it is easy to see that~\C{eqn:S}, which maps $A\leftrightarrow1/A$, is nothing more than T-duality applied to both circles of the torus. So indeed, string theory is invariant under the modular transformations
$$
\G:\varrho\mapsto \frac{a\varrho +b}{c\varrho +d},\quad \begin{pmatrix} a & b\\ c& d\end{pmatrix}\in SL(2,\Z)
$$
of the complexified K\"ahler modulus, $\varrho=b+\I A$. Note that in order to obtain this result, we had to make two departures from conventional geometry: we had to introduce the $B$-field in order to complexify the K\"ahler modulus, and we had to allow for T-duality to get the correct modular transformations. Both of these generalizations arise very naturally in string theory.

The full moduli space of the theory is
$$
\cM(T^2) = \cM_K(T^2)\times \cM_{cs}(T^2) = \left(\G\backslash \IH_\varrho\right) \times \left(\G\backslash\IH_\t\right).
$$
The mirror manifold is another torus, $\check{T}^2$, with the moduli interchanged: $\check\varrho =\t$ and $\check\t =\varrho$. The mirror torus $\check{T}^2$ can be obtained from the original torus, $T^2$, by performing T-duality along one of the circle factors. So in complex dimension 1, mirror symmetry {\it is} T-duality. This statement was conjectured in~\cite{syz}\ to hold in higher dimensions. Roughly speaking the {\it SYZ conjecture} states that if $(X,Y)$ are a mirror pair of Calabi-Yau $n$-folds, then $X$ and $Y$ admit $n$-torus fibrations over a common base $B_n$ of real dimension $n$. Furthermore, the generic fiber in $X$ is a torus $T^n$ and the generic fiber of $Y$ is the T-dual torus $\check T^n$, obtained from $T^n$ by performing T-duality along each of the $n$ circle factors. It is now understood that the SYZ conjecture holds only in certain limits, and does not capture all of the effects of mirror symmetry. See~\cite{Gross:2012syz}\ for more details.

One final comment about mirror symmetry of the torus, which has a nice generalizes to higher dimensions. Consider the Hodge diamond of the torus:
$$
\begin{array}{ccc} \ & h^{00} & \ \\ h^{10} & \ & h^{01} \\ \ & h^{11} & \ \end{array} \quad=\quad \begin{array}{ccc} \ & 1 & \ \\ 1 & \ & 1 \\ \ & 1 & \ \end{array},
$$
where $h^{pq} = {\rm dim} H^{p,q}(T^2,\CC)$ are the Hodge numbers of the torus. This diamond has a fairly large symmetry group, namely the dihedral group $D_4$, and the reflections through various axes have important geometric interpretations.
Two of these are well-known: left-right symmetry corresponds to complex conjugation, while vertical reflection follows from Poincar\'e duality.
What is perhaps less familiar is that reflection through the diagonals corresponds to mirror symmetry.
Recall that the complexified K\"ahler modulus is given by $\varrho=\int(B+\I\omega)$ where $B+\I\omega\in H^{1,1}(T^2,\CC)$.
Similarly, the complex structure modulus can be written $\t=\oint \D z$, where the integration is taken along one of the non-trivial circles.\footnote{More precisely, the integration cycle is (anything homologous to) the closed path that, when lifted to the covering space $\CC$, connects the points $z=0$ and $z=\tau$.}
Thus, mirror symmetry is effectively the interchange of $\D z \D\bar z \leftrightarrow \D z$, or equivalently $H^{1,1}(T^2,\CC) \leftrightarrow H^{1,0}(T^2,\CC)$.
In higher dimensions, a similar statement holds. For a Calabi-Yau $n$-fold $X$, the K\"ahler deformations are parameterized by elements of $H^{1,1}(X,\CC)$ while the deformations of complex structure are parameterized by $H^{1,n-1}(X,\CC)$.
For example in the case of greatest interest, a simply connected 3-fold, the Hodge diamond can be written as
$$
\begin{array}{ccccccc}
\ & \ & \ & h^{00} & \ & \ & \ \\
\ & \ & h^{10} & \ & h^{01} & \ & \ \\
\ & h^{20} & \ & h^{11} & \ & h^{02} & \ \\
h^{30} & \ & h^{21} & \ & h^{12} & \ & h^{03} \\
\ & h^{31} & \ & h^{22} & \ & h^{13} & \ \\
\ & \ & h^{32} & \ & h^{23} & \ & \ \\
\ & \ & \ & h^{33} & \ & \ & \ \\
  \end{array}
   \quad=\quad
\begin{array}{ccccccc}
\ & \ & \ & 1 & \ & \ & \ \\
\ & \ & 0 & \ & 0 & \ & \ \\
\ &0 & \ & h^{11} & \ & 0 & \ \\
1 & \ & h^{12} & \ & h^{12} & \ & 1 \\
\ & 0 & \ & h^{11} & \ & 0 & \ \\
\ & \ & 0 & \ & 0 & \ & \ \\
\ & \ & \ & 1 & \ & \ & \ \\
  \end{array} ,
$$
where we have used the existence of a unique holomorphic top-form, simply-connectedness, conjugation, and duality to reduce the diamond down to two independent numbers: $h^{11}$ and $h^{12}$. So we see that if $(X,Y)$ are a mirror pair of Calabi-Yau 3-folds, then their Hodge diamonds are related by reflecting through the diagonal.
In particular, $h^{11}(X)=h^{12}(Y)$ and $h^{12}(X)=h^{11}(Y)$.

\section{Introduction to Conformal Field Theories}\label{sec:conf}
In the previous section, we saw how string theory can treat very different geometries as equivalent. The underlying reason is that the bizarre looking geometric symmetries, when examined from the worldsheet of the string, amount to simple automorphisms of the conformal field theory (CFT) that describes its dynamics.
Thus, in order to properly understand mirror symmetry, at least in the context of string theory, we must ultimately develop some understanding of CFTs.
In Sect.~\ref{sec:SCFT}, we will specialize to the case of CFTs with $N=(2,2)$ supersymmetry, which is the context in which mirror symmetry was originally discovered.
However, in this section we will introduce some of the basic framework that underlies all CFTs.
Section~\ref{ss:QFT}\ will review a few of the basic features of quantum field theories, of which CFTs form a special class.
After that, we will study the general structure of conformal symmetries in Sect.~\ref{ss:conformal}. A slight generalization of this structure will lead us to the Virasoro algebra in Sect.~\ref{ss:Virasoro}, which underlies all CFTs in two dimensions.
In the following sections,~\ref{ss:Local}\ and~\ref{ss:reps}, we examine how the Virasoro algebra acts on local operators and its representations.
We wrap up this discussion with a several simple, yet important, examples of CFTs to illustrate the formalism in Sect.~\ref{ss:examples}.
Many wonderful references exist which cover this material in greater detail: two-dimensional CFTs were largely developed in the seminal work~\cite{Belavin:1984vu}, while~\cite{Ginsparg:1988ui}\ and Chapter 2 of~\cite{Polchinski:1998rq}\ provide excellent summaries of these results.
The textbook~\cite{DiFrancesco:1997nk}\ has become the gold standard in CFT fundamentals, and contains a wealth of information on this topic.

\subsection{Lightning Review of Quantum Field Theories}\label{ss:QFT}
Before we can properly address conformal field theories, we require some basic knowledge of quantum field theory in general.
This is an extremely vast subject, for which there is no possible way that we can even remotely do justice is just a few short pages.
However, we would like to introduce a few basic concepts to give us a starting point, and also which serve to illustrate some of the key differences between CFTs and the generic quantum field theories.
There are many standard references for this rich topic, which develop the ideas we are about to rush through in much greater detail (and probably with a much clearer presentation).
In particular, for a more mathematical approach to this subject, consult~\cite{Deligne:1999qp,Hori:2003ic}.

Quantum field theory is the successful merger of the two fundamental pillars of modern physics: quantum mechanics and special relativity.
Quantum mechanics controls physics at the smallest scales, from atoms and molecules all the way down to the point-like sub-atomic particles (electrons, quarks, photons, etc).
To a mathematician, quantum systems are rather appealing since they are formulated in purely algebraic terms.
\begin{definition}\label{def:quantum}
To any quantum mechanical system, we assign a Hilbert space $\cH$ called the {\it space of states}.
A {\it state}, which represents a possible configuration of the system, is a ray $\ket{\psi}$ in $\cH$.
An {\it observable} is represented by a Hermitian operator acting on $\cH$, while a {\it symmetry} of the system is represented by a unitary operator acting on $\cH$.
\end{definition}
Observables are measurable quantities (like position or momentum), and the outcome of such a measurement can only be one of the observable's eigenvalues.
A hallmark feature of quantum systems is that the results of these measurements can only be predicted probabilistically.
The most fundamental observable in any quantum system is the energy, whose associated Hermitian operator is called the {\it Hamiltonian}, $H$.
The Hamiltonian has such a prominent role because it controls the time-evolution of states, via Schr\"odinger's equation:
$$
\I\frac{\del}{\del t} \ket{\psi} = H\ket{\psi}.
$$
When $H$ is time-independent, then the Schr\"odinger equation can be integrated directly to determine the time-evolution of any state:
$$
\ket{\psi(t)} = \E^{-\I H (t-t_0)}\ket{\psi(t_0)},
$$
where the unitary operator $U(t,t_0) = \E^{-\I H(t-t_0)}$ is associated with the time-translational symmetry of the system.

The problem with quantum mechanics, and the reason that quantum field theory is unavoidable, is that it is not compatible with special relativity.
If at some initial time, $t_0$, a particle is located at some initial position $\vec{x}_0$, then there typically is a non-zero probability that at {\it any} later time, $t_0+\varDelta t$, the particle can be found at {\it any} other point in space, no matter how small $\varDelta t$.
This contradicts one of the basic principles of special relativity, namely that nothing can propagate faster $c$, than the speed of light.
For slow moving particles (relative to $c$), quantum mechanics serves as a suitable approximation to reality, but at extremely high velocities, namely those close to $c$, relativistic effects become important.
Einstein's great insight into the nature of space and time is that they are not independent, but instead mix under changes in an observer's velocity.
Rather than treating space and time separately, as in quantum mechanics, they should be combined together into a single object: spacetime.
Mathematically, this translates into the statement that space and time comprise a single metric space of indefinite signature, which (in the absence of gravity) should be flat.
\begin{definition}
  {\it $d$-dimensional Minkowski spacetime} is the vector space $\R^{1,d-1}\simeq(\R^{d},\eta_{\m\n})$, equipped with the flat metric $\eta_{\m\n}$ of signature $(1,d-1)$.
\end{definition}
By an appropriate choice of coordinates, $\eta_{\m\n}$ can always be put into the diagonal form $\eta_{\m\n}=diag(-1,1,1\dots,1)$.
Later we will be interested only in $d=2$, but for now we keep the dimension of spacetime arbitrary.
The basic symmetry that underlies special relativity is the {\it Lorentz group}, $SO(1,d-1)$ of (generalized) rotations in spacetime,\footnote{Sometimes these generalized rotations are split into spatial $SO(d-1)$ rotations and ``boosts" along each of the $d-1$ spatial directions.},
$$
x^\m\rightarrow M^\m{}_\n x^\n,
$$
for $M\in SO(1,d-1)$. The full isometry group of Minkowski spacetime is the {\it Poincar\'e group}, which combines the Lorentz group with the set of translations in spacetime
$$
x^\m\rightarrow x^\m+a^\m.
$$
into the semi-direct product $SO(1,d-1)\ltimes \R^{d}$. The orbits of the Lorentz group fall into three classes, depending on whether the squared-distance between a point $(t,\vec{x})\in\R^{1,d-1}$ and the origin,
$$
\D s^2 = \sum_{\m,\n}\eta_{\m\n}dx^\m dx^\n = -dt^2 + d\vec{x}\cdot d\vec{x},
$$
is positive ({\it space-like}), negative ({\it time-like}), or null.
Time-like separated points can always be connected by paths that never exceed the speed of light, while space-like separated points can only be reached by (unphysical) faster than light travel.\footnote{Null separated points can only be reached by traveling at exactly $c$.}
In particular, interactions between objects can only occur {\it locally}, \ie\ when they are at the same point in spacetime.
Time-like separated objects cannot interact directly and must communicate through an intermediary field, such as the gravitational or electro-magnetic fields.
Disturbances in these fields, caused by the local interactions with objects, propagate to future time-like separated points, where they can locally interact with distant objects.

Therefore, by combining quantum mechanics with relativity, the natural observables that emerge in quantum field theory are operator-valued (or quantum) fields.
More precisely, a quantum field is an operator-valued distribution, which can be integrated against test functions to generate an infinite number of (conventional) operators.
\begin{definition}
 In a quantum field theory, a {\it quantum field} (or a {\it local operator}) $\cO(t,\vec{x})$ is an operator-valued distribution defined at each point in spacetime, $(t,\vec{x})\in\R^{1,d-1}$.
 Furthermore, every local operator can be decomposed into an infinite number of {\it mode operators} $\cO_{\vec{n}}$ associated to its each of its (spatial) Fourier modes.
\end{definition}
Quantum fields are defined and act locally at each point in spacetime, so that two local operators at space-like separated points will always commute.
The excitations of these quantum fields, which are responsible for transmitting interactions between distant objects, are interpreted as particles.
The action of a quantum field operator on a state in the Hilbert space is to create or destroy a particle (of the associated type) at the specified point in spacetime.

There are several (equivalent) approaches to quantum field theory, each with their own strengths and weaknesses.
One particularly simple approach is the method of {\it canonical quantization}, which begins by determining the Hamiltonian of the system.
The main advantage of this approach is that it generalizes directly the well-known procedure used in quantum mechanics, but the disadvantage is that Lorentz invariance is not manifest since we must choose a preferred time-like vector in spacetime.
Having chosen a Hamiltonian, we then divide the mode operators of fields into those which raise the energy of a state and those which lower the energy.
Operators that leave a state's energy unchanged necessarily commute with $H$, and so correspond to conserved charges.
We define a lowest energy state, called {\it the vacuum}, to be one that is annihilated by all of the lowering operators.
Then, we build the states of the Hilbert space by acting on the vacuum with all possible raising operators in all possible combinations.
We assign the vacuum zero energy, and measure a state's energy relative to the this.
If follows, by the Lorentz invariance of the theory, that the vacuum be invariant under the entire Poincar\'e group.

\subsubsection*{An Example: Free Scalar Field in $d=4$}
The standard example of a quantum field theory is a free, real, scalar field $\phi(x)$ of mass $m$, and we will focus on $d=4$ for concreteness.
This theory is governed by the action
$$
S[\phi] = \hlf\int \D^4x\left((\dot\phi)^2 -(\vec{\nabla}\phi)^2 -m^2 \phi^2\right),
$$
where $\dot\phi = \del\phi/\del t$.
The field operator $\phi(x)$ can be decomposed into its Fourier modes as
$$
\phi(t,\vec{x}) =  \int\frac{\D^{3}\vec{p}}{(2\pi)^{3}}\frac{1}{\sqrt{2E_{\vec{p}}}} \left(a_{\vec{p}} \E^{-\I p\cdot x} + a^\dagger_{\vec{p}} \E^{\I p\cdot x}\right) ,
$$
where $p\cdot x $ is a Lorentz invariant product
$$
p\cdot x = \sum_\m \eta_{\m\n} p^\m x^\n = -E_{\vec{p}} t + \vec{p}\cdot\vec{x},
$$
and
$$
E_{\vec{p}}^2 = \vec{p}^2 +m^2
$$
is the relativistic energy of the scalar particle. The mode operators, $a_{\vec{p}}$ and $a^\dag_{\vec{p}}$ obey the commutation relations of an infinite set of harmonic oscillators:
\ban
[a_{\vec{p}},a^\dag_{\vec{q}}] &= (2\pi)^3\dd^{3}(\vec{p} - \vec{q}), \\
[a_{\vec{p}},a_{\vec{q}}] &=[a^\dag_{\vec{p}},a^\dag_{\vec{q}}]  =0.
\ean
The advantage of working in this basis is that the Hamiltonian is diagonal:\footnote{The second term in brackets is proportional to $\dd^3(0)$, and therefore infinite, but all is not lost.
These sorts of infinities arise in many problems in quantum field theories, and they are a signal of our ignorance about physics at extremely short distance scales.
Nevertheless, it is well-understood how to regulate and remove these infinite quantities from physically observable quantities.
In this case, since we can only measure energy differences (with respect to the vacuum) the resolution is to simply drop this infinite vacuum energy.}
$$
H = \hlf\int \frac{\D^{3}\vec{p}}{(2\pi)^3E_{\vec{p}}} \left(a^\dag_{\vec{p}} a_{\vec{p}} + \hlf[a_{\vec{p}},a^\dag_{\vec{p}}]\right).
$$
It is easy to check that
$$
[H,a^\dag_{\vec{p}}] = E_{\vec{p}} a^\dag_{\vec{p}},\quad [H,a_{\vec{p}}] = -E_{\vec{p}} a_{\vec{p}},
$$
and so we associate $a^\dag_{\vec{p}}$ with a rasing operator, which creates a particle with momentum $\vec{p}$, and $a_{\vec{p}}$ with a lower operator, which destroys a particle of momentum $\vec{p}$.
In particular, the vacuum is the state annihilated by all $a_{\vec{p}}$, and therefore does not contain any particles.

\subsection{Conformal Groups in Various Dimensions}\label{ss:conformal}
Having recalled some basic facts about general quantum field theories, let us now focus on the conformally invariant ones.
Simply put, a conformal field theory is a quantum field theory where instead of the Poincar\'e group, the underlying symmetry group in the conformal group of spacetime.
Let $(M,g_{\m\n})$ be a (pseudo-)Riemannian manifold of dimension $d$. Recall that under a general coordinate transformation, $x^\mu \rightarrow x'^\m$, the metric tensor is conjugated by the (inverse) Jacobian:
$$
g'_{\m\n}(x') =\sum_{\varrho,\s}  g_{\varrho\s}(x)\frac{\del x^\varrho}{\del x'^\m}\frac{\del x^\s}{\del x'^\n},
$$
so that the infinitesimal line element on $M$, $$\D s^2 = \sum_{\m,\n}g_{\m\n}(x) dx^\m dx^\n,$$ is preserved.
\begin{definition}
The {\it conformal group of $(M,g_{\m\n})$} is the set of all invertible coordinate transformations, $x^\m\mapsto x'^\m$, that leave the metric tensor invariant up to an overall rescaling:
$$
g'_{\m\n}(x') = \La^2(x)g_{\m\n}(x).
$$
\end{definition}
As the name suggests, conformal transformations preserve angles, but not necessarily lengths.
Clearly the set of all isometries of $(M,g_{\m\n})$, which leave the metric invariant, form a subset of the conformal group (with $\La(x)=1$).
Therefore, for Minkowski spacetimes the conformal group contains the Poincar\'e group,
$$
x^\m\rightarrow M^\m{}_\n x^\n,\qquad x^\m\rightarrow x^\m+a^\m,
$$
with $M\in SO(1,d-1)$. Another obvious set of conformal transformation come from {\it dilations}:
$$x^\m\rightarrow x'^\mu = \La^{-1}x^\m,$$
for constant scale factors $\La>0$. A final set of well known angle-preserving transformations come from inversions,
$$
x^\m \rightarrow x'^\m = \frac{x^\m}{x^2},
$$
but these transformations are discrete and we are seeking a set of continuous transformations. The solution is to follow the inversion map by a translation and then another inversion, so the net effect is
$$
x^\m\rightarrow x'^\m = \frac{x^\m+b^\m x^2}{1+2b\cdot x +b^2 x^2}.
$$
This defines the set of {\it special conformal transformations (SCTs)}, which can also be written as
$$
\frac{x'^\m}{x'^2} = \frac{x^\m}{x^2}+b^\m.
$$
It is not hard to show that, at least for $d>2$, this gives the complete list of possible conformal transformations.
\begin{proposition} \label{prop:gens}For $d>2$, the conformal group of $(\R^{d},\eta_{\m\n})$ is generated by the differential operators:
\ban
P_\m =& -\I\del_\m \qquad\qquad\qquad\quad (translations)\\
D =& -\I x^\m\del_\m \qquad\qquad\qquad (dilations)\\
J_{\m\n} =& \I(x_\m\del_\n - x_\n\del_\m) \qquad\quad\  (rotations) \\
K_\m =& -\I(2x_\m x^\n\del_\n -x^2 \del_\m)\ \ (SCTs).
\ean
\end{proposition}
\begin{proof}
\smartqed
The proof follows by considering infinitesimal transformations $x^\mu \rightarrow x^\m+\e\ve^\m(x)$, and solving the conformal Killing equations: $\del_\m\ve_\n +\del_\n\ve_\m = f(x)\eta_{\m\n}$. Taking traces and derivatives of the Killing equation, one can show that $\ve_\m$ can be at most quadratic in $x$ which leads to (the infinitesimal forms of) the transformations listed above. Details can be found on $p.~96$ of~\cite{DiFrancesco:1997nk}.
\qed
\end{proof}
Of course our real interest is when $d=2$, precisely the one exception for which this classification does not apply.\footnote{Although $d=1$ is also excluded, the notion of a conformal transformation is meaningless since every vector is necessarily parallel.}
Nevertheless, we will see that the generators above yield an important subgroup of the full conformal group when $d=2$.
Before exploring the conformal group of $(\R^{2},\eta_{\m\n})$ in detail, let us make a few more remarks regarding the general case.
\begin{proposition}\label{prop:comm}
  For $d>2$, the conformal group of $(\R^d,\eta_{\m\n})$ is isomorphic to $SO(2,d)$.
\end{proposition}
\begin{proof}
\smartqed
Given the explicit forms of the conformal generators in Prop.~\ref{prop:gens}, it follows that they satisfy the following algebra:
\ban
[D,P_\m] &= \I P_\m\\
[D,K_\m] &= -\I K_\m\\
[K_\m,P_\n] &= 2\I \eta_{\m\n}D -2\I J_{\m\n}\\
[J_{\m\n},P_\varrho] &= \I (\eta_{\n\varrho}P_\m - \eta_{\varrho\m}P_\n)\\
[J_{\m\n},K_\varrho] &= \I (\eta_{\n\varrho}K_\m - \eta_{\varrho\m}K_\n)\\
[J_{\m\n},J_{\varrho\s}] &= \I \left( \eta_{\n\varrho}J_{\m\s} + \eta_{\m\s}J_{\n\varrho} - \eta_{\m\varrho}J_{\n\s} - \eta_{\n\s}J_{\m\varrho}\right),
\ean
with all other commutators vanishing. Now let $\m\in\{0,1,2,\dots,d-1\}$ and define
$$
J_{-1,d} =D,\quad J_{-1,\m} = \tfrac{1}{2}(P_\m - K_\m), \quad J_{d,\m} = \tfrac{1}{2}(P_\m + K_\m).
$$
Then the conformal symmetry algebra can be written as
$$
[J_{ab},J_{cd}] = \I \left( \eta_{bc}J_{ad} + \eta_{ad}J_{bc} - \eta_{ac}J_{bd} - \eta_{bd}J_{ac}\right),
$$
where $a,b\in\{-1,0,1,\ldots,d\}$ and $\eta_{ab}=diag(-1,-1,1,1\ldots,1)$. This is the ${\frak so}(2,d)$ algebra.
\qed
\end{proof}
Finally, despite the fact that spacetime actually has a Lorentzian signature, in practice physicists often like to cheat and pretend that it is Euclidean by performing a so-called {\it Wick rotation}.
This amounts to an analytic continuation sending  $x^0\rightarrow \I x^d$, so that the line element
$$
\D s^2 = \sum_{\m,\n=0}^{d-1}\eta_{\m\n}dx^\m dx^\n \ \rightarrow \ \sum_{\m,\n=1}^d \dd_{\m\n}dx^\m dx^\n
$$
becomes effectively Euclidean. We will follow this convention throughout the rest of these notes. After analytic continuation, the conformal group of $(\R^d,\dd_{\m\n})$ is \nobreak{$SO(1,d+1)$}. Now we can proceed to study how the conformal group is modified in $d=2$.

\subsubsection*{Two Dimensions}
As noted earlier, the complete enumeration of conformal generators given in Prop.~\ref{prop:gens}\ only holds for $d>2$. In attempting the same proof for $d=2$, instead of finding that $\ve_\m(x)$ can be at most quadratic in $x$, one finds that $\ve_\m(x)$ must be a harmonic function.
Regarding $\R^2\simeq \CC$ with coordinates $(z,\bar z)$, this is just a reflection of the well known fact that {\it any} holomorphic function $f(z)$ generates a conformal transformation on $\CC$. Under $z\rightarrow z' =f(z)$, independent of $\bar z$,
$$
\D s^2 =\D z \D\bar z \rightarrow \left|\del_z f\right|^2 \D z \D\bar z,
$$
so this is indeed a conformal transformation. In general, such conformal transformations act locally and can only be defined in some open neighbourhood $U\subset\CC$. We will return to the question of global conformal transformations momentarily. Thus, in $d=2$ the {\it local} conformal group is infinite dimensional. We can represent the generators by
$$L_n = -z^{n+1}\frac{\del}{\del z},$$
for all $n\in\Z$, together with their anti-holomorphic partners $\tilde L_n$. Therefore, the conformal group\footnote{Unless we specify a global condition, we will now take the conformal group on $\CC$ to mean the local one.} on $\CC$ is generated by the set of all holomorphic (and anti-holomorphic) vector fields on $\CC^*$. It is easy to see that these generators satisfy the {\it Witt algebra}:
$$[L_m,L_n]=(m-n)L_{m+n},$$
and similarly for $\tilde L_n$.

\subsection{The Virasoro Algebra}\label{ss:Virasoro}
The basic algebraic structure underlying CFTs in two dimensions is the Virasoro algebra, which (together with its supersymmetric generalizations) will play a central role in the rest of these notes.\footnote{While the Virasoro algebra strongly constrains the structure of every CFT, the set of allowed fields and operators in a given theory must obey additional constraints such as locality and modular invariance. See~\cite{Polchinski:1998rq}\ for further details.}
\begin{definition}
  The {\it Virasoro algebra} is the central extension of the Witt algebra by a {\it central charge} $c$. The generators  $L_n$, $n\in\Z$, and $c$ satisfy the commutation relations
  $$
  [L_m,L_n] = (m-n)L_{m+n} +\frac{c}{12}m(m^2-1)\dd_{m,-n}.
  $$
 The central charge $c\in\R$ is a constant that commutes with all of the $L_n$.
\end{definition}
Physically, in a CFT the central charge plays several (related) roles. Firstly, it ``counts" the number of degrees of freedom in the theory. Secondly, it governs the response of a conformal theory to the introduction some length scale.\footnote{For example, we can map the complex plane to a cylinder of radius $R$ by the conformal mapping
$
z\rightarrow w = R \log z.
$
However, in doing we break the scale invariance of the system by introducing the preferred length $R$, and this is reflected by a change in the vacuum energy density by an amount $-{c/(24R^2)}$.}
Thirdly, it measures the breakdown of conformal invariance when a CFT is placed on a curved surface, such as $\P^1$, instead of $\CC$.
Notice that the generators $L_0,L_{\pm1}$ satisfy a closed $SL(2,\R)$ sub-algebra:
$$[L_0,L_{\pm1}] =\mp L_{\pm1},\quad [L_1,L_{-1}]=2L_0,$$
independent of the central charge. This is the {\it global} conformal group, which maps all of $\CC$ to itself. Together with anti-holomorphic generators $\tilde L_0,\tilde L_{\pm1}$ these combine into
$$ SL(2,\R)\times SL(2,\R)\simeq SL(2,\CC)\simeq SO(1,3).$$
Thus, the global portion of the conformal group in $d=2$ takes the same form as in higher dimensions, namely $SO(1,d+1)$.

Rather than dealing with an infinite number of generators, it is convenient to package the $L_n$ into a single local operator that generates all possible conformal transformations.
\begin{definition}
 The {\it energy-momentum tensor}, $T(z)$, is the local operator that generates the complete set of (local) conformal transformations. The relation between $T(z)$ and the  Virasoro generators is given by the formal Laurent series
 $$
 T(z) = \sum_{n\in \Z} \frac{L_n}{z^{n+2} }.
 $$
\end{definition}
We can recover the Virasoro algebra from the energy momentum operator by studying the behaviour of the product $T(z)T(w)$, as $z\rightarrow w$. The main tool for this analysis is the {\it operator product expansion (OPE)}, which we will properly define later in Def.~\ref{def:OPE}. For now, we will simply assert the following:
\begin{proposition}\label{fact:TT}
In any CFT, the OPE of the energy-momentum tensor with itself is given by
\be
T(z)T(w) \sim \frac{c/2}{(z-w)^4} +\frac{2T(w)}{(z-w)^2} + \frac{\del T(w)}{z-w},\label{eqn:TT}
\ee
where $\sim$ denotes equivalence up to non-singular terms in the limit $z\rightarrow w$.
\end{proposition}
Deriving this result from first principles would require introducing much more CFT formalism that we intend to cover here, and interested readers should consult the references listed at the beginning of this section.
 We will substantiate this claim by considering specific examples in Sect.~\ref{ss:examples}. However, given this assertion we can recover the Virasoro algebra, as claimed.
\begin{proposition}\label{prop:TT}
The singular terms in the OPE~\C{eqn:TT}\ imply the Virasoro algebra for the mode operators $L_n$.
\end{proposition}
\begin{proof}
\smartqed
The idea is that we can always recover the Virasoro generators by taking appropriate residues of $T$:
$$
L_n = \oint \frac{\D z}{2\pi\I}\, z^{n+1}T(z).
$$
Logically, the commutator $[L_m,L_n]$ requires taking two contour integrals of the product $T(z)T(w)$, but the two contours should be linked somehow, otherwise we would only recover the product $L_m L_n$, instead of the commutator.
To obtain the two orderings of the mode operators, we must consider the cases where $|z|<|w|$ and $|z|>|w|$.
We accomplish this ordering by taking the $z$ contour around the point $w$, and then the $w$ contour around, say, the origin. So, in terms of local operators, the prescription to obtain commutators of mode operators is
$$
[L_m,L_n] = \oint_0\frac{\D w}{2\pi\I} \oint_w\frac{\D z}{2\pi\I}\, z^{m+1} T(z) w^{n+1} T(w),
$$
where the subscripts on the integrals indicates the point about which we integrate. Inserting the $TT$ OPE and computing the residues, we obtain
\ban
[L_m,L_n] &= \oint_0\frac{\D w}{2\pi\I} w^{n+1}\oint_w\frac{\D z}{2\pi\I} z^{m+1}\left[\frac{c/2}{(z-w)^4} +\frac{2T(w)}{(z-w)^2} + \frac{\del_w T(w)}{z-w} +\ldots \right]\\
&= \oint_0\frac{\D w}{2\pi\I}  \left[\frac{c}{12}m(m^2-1)w^{n+m-1} + 2T(w)(m+1)w^{m+n+1} + \del T(w)w^{m+n+2} \right]\\
&= \frac{c}{12}m(m^2-1) \dd_{m+n,0} +(m-n)L_{m+n},
\ean
as required.
\qed
\end{proof}
Thus we have two equivalent ways to think about the Virasoro algebra: either in terms of its generators, $L_n$, or in terms of the energy-momentum operator, $T(z)$. One should not be fooled into thinking that we have somehow replaced and infinite number of operators by just a single one, since $T(z)$ defines an operator at each point $z\in\CC$. As the Laurent expansion relating the two clearly demonstrates, neither one contains more information that then other.
As we saw in Sect.~\ref{ss:QFT}, dual presentations of this sort are prevalent in the study of quantum field theories, and are especially useful in the understanding of CFTs.

\subsection{Local Operators}\label{ss:Local}
So far the only local operator we have considered in a CFT is the energy momentum operator $T(z)$.
Let us now discuss some of the general properties common to all local operators that appear a generic CFT.
\begin{definition}
  A local operator $\cO(z,\bar z)$ has {\it weights} $(h,\tilde h)$ if, under a global rescaling of the coordinates, it transforms according to
  $$
  \cO'(\la z,\tilde\la\bar z) = \la^{-h} \tilde \la^{-\tilde h}\cO(z,\bar z).
  $$
  The {\it dimension} and {\it spin} of $\cO$ are given by the respective sum and difference of the weights: $\Delta=h+\tilde h$, and $s=h-\tilde h$.
  \end{definition}
Note that $\Delta$ is the eigenvalue of $\cO$ under infinitesimal dilations, which are generated by $L_0+\tilde L_0$, while $s$ is the eigenvalue under infinitesimal rotations, generated by $\I(L_0-\tilde L_0)$.
For notational simplicity, we will usually suppress the $\bar z$ dependence of local operators and simply write $\cO(z)$, keeping in mind that generic local operators will also involve an anti-holomorphic sector.
A key tool in the analysis of CFTs is the operator product expansion, which we will now define for arbitrary local operators.
\begin{definition}\label{def:OPE}
  Let $\{\cO_k(z)\}$ be the complete set of independent local operators appearing in a given CFT. The {\it operator product expansion (OPE)} of two local operators, $\cO_i(z)$ and $\cO_j(w)$, relates their product to a (possibly infinite) sum of local operators:
  $$
  \cO_i(z) \cO_j(w) = \sum_k c_{ijk}\, (z-w)^{h_k -h_i -h_j} \cO_k(w),
  $$
  where $h_i$, $h_j$, and $h_k$ are the weights of the corresponding local operators, and the structure coefficients $c_{ijk}$ are constants.
\end{definition}
Note that the form of the righthand side is completely fixed by global conformal invariance. However, symmetry alone cannot constrain the values of the constants $c_{ijk}$, which are analogous to the structure constants of a Lie algebra.
The OPE defines a convergent series within a radius set by the distance to the nearest local operator. For example, given a triple product of local operators, $\cO_1(z_1) \cO_2(z_2) \cO_3(z_3)$, if we expand $\cO_1 \cO_2$ about the point $z_2$, then this OPE converges within a radius of $|z_2-z_3|$.
A local operator's weights, together with the general behaviour of fields under translations, provides sufficient data to determine the following behaviour.
\begin{proposition}\label{prop:TO}
  Let $\cO(z)$ be a local operator of weight $h$. Then, the singular terms of the $T\cO$ OPE take the universal form
  \ban
  T(z)\cO(0) &\sim \ldots +\frac{h\, \cO(0)}{z^2} + \frac{\del \cO(0)}{z},
  \ean
  where the $\ldots$ denote higher order poles, which depend upon the choice of $\cO(z)$. In particular, $T(z)$ has weight $h=2$.
\end{proposition}
\begin{proof}
\smartqed
  The key point is that under an infinitesimal conformal transformation,
  $$
  z\rightarrow z' = z+\sum_{n\in\Z} z^{n+1}\ve_n,
  $$
  where $\ve_n$ are a collection of small parameters, the variation of $\cO(z)$  is given by its commutator with the appropriate generator:
  $$
  \dd\cO(z) \equiv \cO'(z) - \cO(z) = -\sum_n \ve_n [L_n,\cO(z)] +O(\ve^2).
  $$
  Similar to the method employed in Prop.~\ref{prop:TT}, we can write these variations in terms of contour integrals:
  $$
  \dd_n\cO(0) = -\ve_n[L_n,\cO(0)] = -\ve_n\oint \frac{\D z}{2\pi\I} \, z^{n+1} T(z) \cO(0).
  $$
The two lowest poles in the $T\cO$ OPE follow immediately from the fact that $L_0$ generates (holomorphic) dilations, with eigenvalue $h$, and $L_{-1}$ generates translations in $z$. In particular, under a translation we have $\cO'(z) =\cO(z-\e)$
and so
$$
\dd_{-1}\cO(z) = -\e[L_{-1},\cO(z)] = \cO'(z) -\cO(z)  =  -\e\del \cO(z).
$$
The weight of $T(z)$ can then be read off from the $TT$ OPE in~\ref{eqn:TT}.
\qed
\end{proof}

From a single local operator, we can extract conventional operators by taking appropriate contour integrals, just as we did for $T(z)$.
\begin{definition}
Given a local operator $\cO(z)$ of weight $h$, we associate an infinite number of {\it mode operators}, denoted $\cO_n$, by the relations
$$
\cO(z) = \sum_{n\in\Z} \frac{\cO_{n}}{z^{n+h}},\qquad \cO_n =\oint \frac{\D z}{2\pi\I} \,z^{n+h-1} \cO(z).
$$
\end{definition}
The singular terms in the OPE of $\cO(z)\cO'(w)$ completely determines the mode algebra $[\cO_m,\cO'_n]$, and vice-versa.

\subsubsection*{Primary Operators}
In Prop.~\ref{prop:TO}, we saw that the lowest order poles in the OPE of $T(z)$ with any local operator $\cO(w)$ is fixed by its behaviour under dilations and translations.
To determine the higher order poles require knowledge of how $\cO$ behaves under the action of $L_n$, for $n\geq1$, which is not universal.
However, there is an important set of local operators for which the poles in the OPE with the energy-momentum tensor are completely determined.
\begin{definition}\label{def:primary}
A local operator $\cO(w)$ is called {\it primary} if its OPE with $T(z)$ has a pole of order (at most) 2 at $z=w$:
$$
T(z)\cO(w) \sim \frac{h\, \cO(w)}{(z-w)^2} +\frac{\del\cO(w)}{z-w},
$$
with all higher order singular terms vanishing. A local operator that is not primary is called {\it secondary} or a {\it descendant}.
\end{definition}
Notice that $T(z)$ is not a primary operator, unless $c=0$. As explained earlier, the central charge is related to the breakdown of conformal invariance in a CFT, and the non-primary nature of $T(z)$ is related to this fact.
It can be shown that the infinitesimal conformal transformations of a primary operator, determined by its OPE with $T$, ``integrates" to the following finite form
$$
\cO'(z') = \left(\del_z z'\right)^{-h} \cO(z)
$$
for {\it any} conformal transformation, not just the global dilations we used to define $h$ and $\tilde h$. Non-primary operators, such as $T(z)$, would have additional correction terms on the right-hand side of the finite conformal transformation.
A nice feature of primary operators is that their algebra with the Virasoro generators is completely determined:
$$
[L_m,\cO_n] = \left(m(h-1)-n\right)\cO_{m+n}.
$$
We will see in the coming sections that primary operators play a central role in CFTs, and their supersymmetric extensions.

\subsection{Representations of the Virasoro Algebra}\label{ss:reps}
For a Euclidean field theory, canonical quantization is a rather ambiguous procedure since there is no ``time" direction to single out.
In order to help us choose a Hamiltonian, it helps to recall that our ultimate goal is understanding the CFTs that live on the worldsheet of a propagating string.
To that end, consider an infinitely long cylinder, ${\cal C}=\R\times S^1$, which we identify as the worldsheet of a closed string.
It is natural to think of the length of the cylinder as corresponding to a ``time" direction, $t\in\R$,  and its circumference as ``space", with $\th\sim\th+2\pi$.
Then, since we are working with a two-dimensional conformal theory, we can consider the conformal map
$$
z = \E^{t+\I \th},
$$
which maps ${\cal C}$ to the complex plane $\CC$.
\begin{figure}[h]
\centering
\includegraphics[trim= 0 400 0 50 ,clip,scale=0.5]{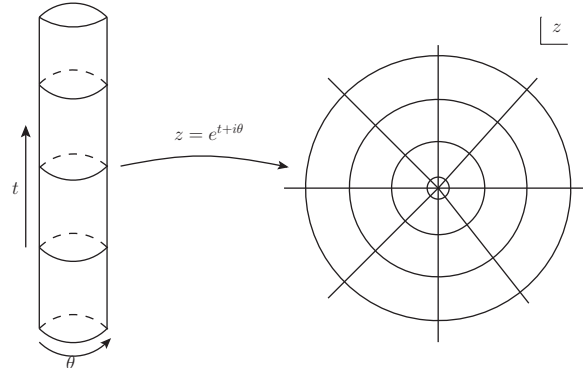}
\caption{\textit{Conformal map from the cylinder, ${\cal C}$, to the complex plane, $\CC$}}
\label{fig:Map}
\end{figure}
In particular, our ``time" direction, $t$, is now identified with the radial direction on the plane, and the infinite past, $t\rightarrow-\infty$, maps to the point $z=0$, while distant future maps to the point $z=\infty$ on $\P^1$.
This characterizes our quantization procedure.
\begin{definition}
  In the {\it radial quantization} of a CFT on $\CC$, the Hamiltonian is chosen to coincide with dilation operator, \ie\  $H=L_0+\tilde L_0$, which generates radial evolution.\footnote{To be precise, changing between cylinder and plane frames induces a shift in $H$ by $(c+\tilde c)/{24}$ because of the anomalous (\ie non-primary) transformation properties of $T(z)$ and $\tilde T(\bar z)$. For simplicity, throughout these notes we will ignore these subtle corrections, and refer the interested reader to the references for a more thorough treatment.} Therefore, a state's energy is given by its dimension $\Delta=h+\tilde h$.
If the spectrum of $H$ is discrete, the theory is called {\it non-degenerate}, otherwise it is {\it degenerate}.
\end{definition}
Now that we have chosen a Hamiltonian, we can separate all of the mode operators into raising and lowering types (and those that generate symmetries). From the Virasoro algebra, we have
$$
[L_0,L_n] = -nL_n,
$$
so that $L_n$ with $n>0$ lowers the energy of a state, while $n<0$ raises its energy. Also for a primary operator $\cO(z)$, we have
$$
[L_0,\cO_n] = -n\cO_n,
$$
so once again modes with $n>0$ act as lowering operators, and $n<0$ are raising operators.\footnote{This is why we shifted the powers of $z$ in the Laurent expansion of $\cO(z)$ by $h$.}
Having identified the lowering operators on the quantum theory, we may now assume the existence of a vacuum.
\begin{definition}
  The {\it vacuum state}, denoted $\ket{0}$, is the $SL(2,\CC)$ invariant vector in the Hilbert space of states that is annihilated by all of the lowering operators of the CFT. In particular,
  $$
  L_n\ket{0} = \tilde L_n\ket{0} = 0,\quad\forall~n\geq-1.
  $$
\end{definition}
Notice that, in analogy with the Poincar\'e invariance of the vacuum in a typical quantum field theory, we have demanded that the vacuum be invariant under the global conformal group, $SL(2,\CC)$.\footnote{We will see by the end of this section that demanding invariance of the vacuum under the full Virasoro algebra is too strong a requirement.}
With some additional mild assumptions, we will show that at the end of this section that the vacuum state is well-defined as the unique (up to scalar multiplication) lowest energy, $SL(2,\CC)$  invariant state.

Starting from the vacuum state, we can build up the entire Hilbert space of states, ${\cal H}$, by acting with all possible combinations of the rasing operators, $L_{-n}$ and $\cO_{-n}$ for $n>0$.
A remarkable fact about CFTs, which certainly does not hold for a typical quantum field theory, is that the structure of the Hilbert space is completely determined by the set of local operators.
\begin{theorem}
  {\rm (The state-operator correspondence)} In any CFT, especially within the framework of radial quantization, the Hilbert space of states is isomorphic to the complete set of local operators.
\end{theorem}
\begin{proof}
\smartqed
  We will only present a heuristic proof of this important theorem, leaving the details to the references. We start by considering an arbitrary state $\ket{\psi}$ in the cylinder frame, ${\cal C}$.
  This state corresponds to the complete set of profiles $\{\cO_i(\th)\}$ of the fields (\ie\ local operators) at a fixed time slice, acting on the vacuum.
  In a string theoretic application, these states then correspond to the possible configurations of the string in space at any given time.
  The time evolution of these states is dictated by the unitary operator $\E^{-\I Ht}$. Thus, for an arbitrary state we have
  $$
\ket{\psi(\th,t)} = \E^{-\I Ht}\ket{\psi(\th)} = \E^{-\I Ht}\sum_i \cO_i(\th)\ket{0}.
$$
The beauty of radial quantization is that if we propagate any state back to $t\rightarrow-\infty$, and apply the conformal map $z=\exp(t+i\th)$, then the entire spatial slice of ${\cal C}$ in the infinite past gets mapped to the point $z=0$.
Thus, the {\it non-local} state $\ket{\psi(\th)}$ is mapped to the local operator $\sum_i\cO_i(0)$.

Going in the other direction, we can begin with any local operator $\cO(z)$ of weight $h$ and consider its mode decomposition,
$$\cO(z) = \sum_n z^{-n-h}\cO_n.$$
Clearly, only modes with $n+h\geq0$ could possibly contribute at $z=0$, and we already know that modes with $n>0$ will annihilate the vacuum. In order that $\cO(0)\ket{0}$ be well-defined, we must postulate that not only $n>0$ annihilate the vacuum but
$$
\cO_n\ket{0}=0,\quad \forall~n\geq1-h.
$$
Notice that this is the behaviour of $T(z)$, which has $h=2$ and $L_{-1}\ket{0}=0$. We will see in examples that this assumptions is indeed satisfied.
Therefore, to the local operator $\cO(z)$ we can assign the state
$$
\ket{\cO} = \lim_{z\rightarrow0}\cO(z)\ket{0} = \cO_{-h}\ket{0}.
$$
In particular, the vacuum state, $\ket{0}$, simply corresponds  to the unit operator $1$.
\qed
\end{proof}
Thus, to understand the structure of the Hilbert space we only need to study the local operators, which have already done in detail.
In particular, recall that primary operators are supposed to play a central role in CFTs. Let us now explain why.
\begin{definition}
A state $\ket{\cO}$ is called a {\it primary state} if it is associated (by the state-operator correspondence) with a primary operator. Equivalently,
$$
L_0\ket{\cO}= h \ket{\cO},\qquad L_n\ket{\cO} = 0,\quad \forall~n>0.
$$
A state $\ket{\cO_{k_1,k_2,k_3,\ldots}}$ is called a {\it descendant of $\ket{\cO}$} if it follows from the primary state $\ket{\cO}$ by application of the raising operators $L_{-n}$:
$$
\ket{\cO_{k_1,k_2,k_3,\ldots}} = \ldots L_{-k_3}L_{-k_2}L_{-k_1}\ket{\cO},\quad 0<k_1\leq k_2\leq k_3 \ldots.
$$
A primary state together with all of its descendants comprise a {\it conformal family}, also called a {\it Verma module}.
\end{definition}
Notice that if a primary state has weight $h$, then its descendants will have weights $h+\sum_i k_i>h$. Also, note that the energy-momentum tensor is a descendant of the vacuum, since $T(z)\ket{0}=L_{-2}\ket{0}$.
Clearly, a conformal family transforms as a representation of the Virasoro algebra\footnote{In most cases of interest, the Virasoro Verma module is not irreducible and the corresponding CFT is neither rational nor unitary. However, one can also construct a CFT associated to the irreducible quotient. This is then a CFT with an interesting representation theory, especially there are central charges for which the simple Virasoro algebra is the symmetry algebra of a unitary rational CFT. A complete book is devoted to this subtle question \cite{KI}. We thank T.~Creutzig for explaining this point to us.}, and furthermore it is completely characterized by its primary, which is the state of lowest weight.\footnote{Ironically, these states are usually referred to as highest weight states, in analogy with standard Lie theory.}
Thus, to understand the entire Hilbert space of a CFT it suffices to study the primary states as the rest are related by conformal transformations.

\subsubsection*{An Aside: Building Representations}
Perhaps the statement that a conformal family forms a representation of the Virasoro algebra deserves further comment. Let us illustrate this point by briefly recalling the construction of (irreducible) representations of ${\frak su}(2)$.
To that end, we begin with the algebra $[J_0,J_\pm]=\pm J_\pm,~[J_+,J_-]=2J_0$. We work in an eigenbasis of $J_0$, so that $J_0\ket{j}=j\ket{j}$.
Then we postulate the existence of a lowest weight state $\ket{j_{\rm min}}$, such that $J_-\ket{j_{\rm min}}=0$. We build our representation of ${\frak su}(2)$ by applying $J_+$ repeatedly to $\ket{j_{\rm min}}$.
However, $J_+$ is typically nilpotent with $(J_+)^{2j_{\rm min}+1}=0$. Thus, each state in the $(2j_{\rm min}+1)$-dimensional representation of ${\frak su}(2)$ is then of the form $(J_+)^n\ket{j_{\rm min}}$ for some $0\leq n\leq 2j_{\rm min}$.
The only differences between the representations of the Virasoro algebra and ${\frak su}(2)$ is that there are an infinite number of raising operators, $L_{-n}$ for $n>0$, and they are no longer nilpotent. Thus, the irreducible representations of the Virasoro algebra (the conformal families) are all infinite dimensional.

\subsubsection*{Unitary CFTs}
Let us close this section by describing some general features of the spectrum of a wide class of (physically relevant) CFTs.
\begin{definition}
  A conformal field theory is called {\it unitary} if the inner product $\braket{\ }{\ }$ on its Hilbert space is positive and
  $$
  L_n^\dag = L_{-n},\quad \tilde L_n^\dag = \tilde L_{-n}.
  $$
\end{definition}
Note that the dual state $\bra{0}$ is naturally identified with the insertion of the unit operator $1$ at the point $z\rightarrow\infty$, corresponding to the infinite future. Our discussions will focus solely on unitary CFTs, in part because they behave as we expect.
\begin{proposition}\label{prop:unitary}
  Let $\cO(z,\bar z)$ be a local operator with weights $(h,\tilde h)$. Then, in any unitary CFT with central charges $(c,\tilde c)$ the following properties hold:
  \begin{itemize}
  \item $h\geq0$, with $h=0\ \Leftrightarrow\ \cO(z,\bar z) = \cO(\bar z)$,
  \item similarly $\tilde h\geq0$, with  $\tilde h=0\ \Leftrightarrow\ \cO(z,\bar z) = \cO(z)$,
  \item and $c,\tilde c\geq0.$
  \end{itemize}
\end{proposition}
\begin{proof}
\smartqed
  To prove the lower bound on $h$, it suffices to consider only the primary states since a descendant's weight is always greater. Then, by the positivity of the inner product we have
  $$
  0\leq\left|\left|L_{-1}\ket{\cO}\right|\right|^2 = \bra{\cO}L_1L_{-1}\ket{\cO} = \bra{\cO}(2L_0+L_{-1}L_1)\ket{\cO} = 2h,
  $$
since $L_n\ket{\cO}=0$ for $\cO$ primary and $n>0$. Note that if $\cO$ is not primary (with respect to $L_n$)\footnote{If $\cO$ is a descendent solely by the action of $\tilde L_{-n}$ then we consider it primary for this argument.},
then $h>h_p\geq0$ where $h_p$ is the smallest weight within its entire conformal family. Therefore, if $h=0$ then $\cO$ must be primary and we have
$$
h=0\ \Rightarrow\ L_{-1}\ket{\cO}=0\ \Leftrightarrow\ \del\cO(0)\ket{0}=0\ \Leftrightarrow\ \del\cO(z,\bar z)=0,
$$
and so $\cO(z,\bar z)=\cO(\bar z)$. However, for {\it any} state $\ket{\psi}$ we have
$$
0\leq 2h_\psi = \bra{\psi}2L_0\ket{\psi} = \bra{\psi}[L_1,L_{-1}]\ket{\psi} = ||L_{-1}\ket{\psi}||^2 - ||L_1\ket{\psi}||^2.
$$
So if $\cO$ is anti-holomorphic then $L_{-1}\ket{\cO}=0$, implying $0\leq h\leq 0$, and therefore $h=0$. As a bonus, we learn that states with $h=0$ are also annihilated by $L_1$, and thus invariant under the holomorphic $SL(2,\R)$.
Analogous arguments apply for $\tilde h$.

Finally, to show the lower bound on the central charges, we compute the norm of $L_{-n}\ket{0}$ for any $n>1$:
$$
0\leq ||L_{-n}\ket{0}||^2 = \bra{0}L_nL_{-n}\ket{0} = \bra{0}\left(2nL_0 + \frac{c}{12}n(n^2-1)\right)\ket{0} = \frac{c}{12}n(n^2-1).
$$
Since $n>1$, we have $c\geq0$, and similarly for $\tilde c$.
\qed
\end{proof}
Some basic properties of the vacuum state (claimed earlier) follow immediately.
\begin{corollary}
In a unitary CFT, the vacuum is the unique $SL(2,\CC)$ invariant state, and furthermore it has the lowest possible energy.
\end{corollary}
\begin{proof}
\smartqed
  The vacuum has $h=\tilde h=0$, which is minimal, and from the previous proof this implies $SL(2,\CC)$ invariance. The only thing to check is uniqueness. Suppose $\ket{\cO}$ is another $SL(2,\CC)$ invariant state. Then $\del\cO=\bar \del\cO = 0$, which means $\cO$ is constant and therefore (a multiple of) the unit operator.
\qed
\end{proof}

\subsection{Examples of Free CFTs}\label{ss:examples}
In this section we will briefly present some examples of simple CFTs.
This will serve to illustrate the formal concepts we have developed so far, and will also serve as the basic building block for more complicated examples when we discuss superconformal theories.
Many results will be stated without proof, since a full detailed account would be too involved. Hopefully, we have developed the general theory up to this point enough that the reader should be able to fill in many of the steps.
Some results, however, require material beyond what we have covered in these notes, and the reader should consult the references.
\subsubsection*{Free Scalar Field}
The standard example to begin with is a single, free (massless) scalar field, $X(z,\bar z)$. The action for this theory is
$$
S = \frac{1}{4\pi}\int \D^2z\ \del X \bar\del X,
$$
which we may interpret as the worldsheet theory of string propagating on $\R$.\footnote{The normalization is set by the tension, $T=({2\pi\a'})^{-1}$. In this section we will work in units where $\a'=2$, which will considerably simply the formulas. This convention differs from Sect.~\ref{ss:Tdual}, where we used $\a'=1$.} Multidimensional generalizations are straightforward to construct by including additional scalar fields.
The Euler-Lagrange equation of motion for this action is simply
$$
\del\bar\del X(z,\bar z)=0.
$$
So, at least classically, $X(z,\bar z)$ is harmonic function.
However, we know that in a quantum mechanical theory $X(z,\bar z)$ should be treated as a local operator, and singularities can develop when local operators approach one another.
It turns out, by a basic quantum field theory computation, that this equation of motion holds up to insertions of $X$ at the same point:
$$
X(w,\bar w)\del\bar\del X(z,\bar z)  = -2\pi \dd^2(z-w,\bar z-\bar w).
$$
This determines the OPE
$$
X(z,\bar z) X(w,\bar w) \sim -\ln|z-w|^2,
$$
up to possible holomorphic/anti-holomorphic (and therefore non-singular) terms.
One lesson the reader should take away from this OPE is that $X$ not a ``good" operator, since it's OPE is non polynomial in $(z-w)^{-1}$.
On the other hand the operators $\del X(z)$ and $\bar\del X(\bar z)$, which by the equation of motion are respectively holomorphic and anti-holomorphic fields, are much better behaved.
For example, the OPE of $\del X$ with itself is
\be\label{eqn:dXdX}
\del X(z) \del X(w) \sim -\frac{1}{(z-w)^2},
\ee
which follows by suitably differentiating the $XX$ OPE. In addition, these are the operators used to construct the energy momentum operators:
$$
T(z) = -\hlf \del X(z) \del X(z),
$$
and similarly for $\tilde T(\bar z)$.\footnote{As a composite local operator, we should be careful in how we define $T(z)$, because of potential singularities coming from the OPE.
Typically, this ambiguity is handled by so-called {\it normal ordering}, by defining composite operators with the singular terms subtracted off.
So, more precisely, the energy momentum operator is given by
$$
T(z) = -\hlf \lim_{w\rightarrow z}\left(\del X(z) \del X(w) +\frac{1}{(z-w)^2}\right).
$$
In what follows, we will always assume that composite operators are normal ordered.}
Notice that $\bar\del T(z) = \del\tilde T(\bar z)=0$, which is a sign that this theory is in fact conformal.
Using~\C{eqn:dXdX}\ we can work out the product rule for $T(z)$ with itself:
\ban
T(z) T(w) &\sim \frac{1/2}{(z-w)^4} -\frac{\del X(z) \del X(w)}{(z-w)^2} \\
&\sim \frac{1/2}{(z-w)^4} -\frac{\del X(z)\del X(z)}{(z-w)^2} - \frac{\del^2 X(z)\del X(z)}{z-w}\\
&\sim \frac{1/2}{(z-w)^4} + \frac{2T(w)}{(z-w)^2} +\frac{\del T(w)}{z-w},
\ean
where the first term on the righthand side comes from the two ways to contract all four $\del X$, and the second term comes from the four possible pairwise products.
This confirms the claim in Prop.~\ref{fact:TT}\ for the $TT$ OPE, at least for this specific example, which turns out to have $c=\tilde c=1$.
Similarly, we can work out
\ban
T(z) \del X(w)  &\sim \frac{\del X(z)}{(z-w)^2}\ \sim\ \frac{\del X(w) }{(z-w)^2}+\frac{\del^2 X(w)}{(z-w)},
\ean
which shows that $\del X$ is a primary with $(h,\tilde h)=(1,0)$.

The Laurent expansion for $\del X(z)$ is holomorphic, and similarly for $\bar\del X(\bar z)$:
$$
\del X(z) = -\I \sum_{n\in\Z} \frac{\a_n}{z^{n+1}},\quad \bar\del X(\bar z) = -\I \sum_{n\in\Z}\frac{\tilde \a_n}{\bar z^{n+1}}.
$$
As usual, the modes $\a_n$ with $n<0$ act as raising operators and $n<0$ as lowering operators.
By taking residues of~\C{eqn:dXdX}, we can easily work out the associated commutators:
\be\label{eqn:aa}
[\a_m,\a_n] = m\dd_{m,-n}.
\ee
$\a_0$ and $\tilde \a_0$ are special since they commute with $L_0$ and $\tilde L_0$, and therefore they correspond to conserved charges.
In fact they correspond to the same charge, namely the momentum in the target space $\R$.
To see this, integrate the Laurent series of $\del X$ and $\bar\del X$:
$$
X(z,\bar z) = x -\I \a_0\ln z -\I \tilde\a_0\ln\bar z +\I \sum_{n\neq0}\frac1n\left(\frac{\a_n}{z^n} +\frac{\tilde \a_n}{\bar z^n}\right),
$$
where $x$ is an integration constant that we may interpret as the center of mass position of the string.
The appearance of $\ln(z)$ reinforces the notion that $X$ is not a ``good" operator.
In order that $X$ remains single-valued as $z\rightarrow \E^{2\pi\I}z$, we must impose the equality
$$
\a_0 = \tilde \a_0=:p.
$$
By taking residues of the $XX$ OPE, in addition to~\C{eqn:aa}, we find
$$
[x,p]=\I .
$$
This is the familiar canonical commutation relation from quantum mechanics, and justifies the identification of the zero-mode, $p$, with the momentum of the string.
By inserting the mode expansion for $\del X$ in $T(z)$, we can determine the Virasoro generators in terms of $\a_n$:
$$
L_m = \hlf\sum_{n\in\Z} \a_{m-n}\a_n.
$$
For $m\neq0$ this is result is unambiguous, since the mode operators in each term commute, but for $m=0$ we must specify the order.
We follow the standard convention of placing the raising operators to the left of the lowering operators, so that
$$
L_0 := \frac{\a_0^2}{2} + \sum_{n>0}\a_{-n}\a_n.
$$
In particular, the Hamiltonian is given by\footnote{Again, we are suppressing a constant additive shift in the spectrum.}
$$
H = L_0 + \tilde L_0 = p^2 + \sum_{n>0}\left(\a_{-n}\a_n + \tilde\a_{-n}\tilde\a_n\right).
$$

A general state in the Hilbert space will be labeled by its momentum, $k$, together with the ``occupation numbers", $N_n$ and $\tilde N_n$, for each of the $\a_{-n}$ and $\tilde\a_{-n}$ modes.
That is, every state can be written in the form
$$
\ket{k;N,\tilde N} =\ket{k;N_1,N_2,\ldots,\tilde N_1,\tilde N_2,\ldots} = \ldots \tilde\a_{-2}^{\tilde N_2}\tilde\a_{-1}^{\tilde N_1}\ldots \a_{-2}^{N_2}\a_{-1}^{N_1} \E^{\I kx}\ket{0}.
$$
The primary states of this theory are those with $N_1,\tilde N_1\leq1$ and $N_n=\tilde N_n=0$, for all $n>1$. In particular, for every $k\in\R$ there exists a primary state $\ket{k;0,0}$ associated with the operator $\E^{\I kX(z,\bar z)}:$
$$
\ket{k;0,0} = \lim_{z,\bar z\rightarrow0}\E^{\I kX(z,\bar z)}\ket{0} = \E^{\I kx}\ket{0}.
$$
It is easy to check that these states are indeed the momentum eigenstates:
$$
p\ket{k;0,0} = p \E^{\I kx}\ket{0} = k \E^{\I kx}\ket{0} = k\ket{k;0,0},.
$$
where we used $[x,p]=\I$ to pull down the factor of $k$.
The spectrum of the theory is similarly easy to compute. For an arbitrary state,
$$
H\ket{k;N,\tilde N} = \left(k^2 + \sum_{n>0} n(N_n+\tilde N_n)\right)\ket{k;N,\tilde N},
$$
which is a slight refinement of what we claimed in Sect.~\ref{ss:Tdual}. In particular, since $\R$ is non-compact there is no quantization of the momentum, and so the spectrum is actually degenerate.

\subsubsection*{Compactified Scalar Field}
In obtain a discrete spectrum, we require a compact target space. This is easy enough to achieve in one dimension by imposing a periodicity in $X$:
$$X\simeq X+2\pi R.$$
In order that the operator $\E^{\I kx}$, which creates a state with momentum $k$, be single-valued under $X\rightarrow X+2\pi R$, the momentum must be quantized:
$$
k = \frac{n}{R},\quad n\in\Z.
$$
Furthermore, as we circle the complex plane the field $X(z,\bar z)$ no longer needs to be single-valued:
$$
X(\E^{2\pi\I}z,\E^{-2\pi\I}\bar z) = X(z,\bar z) +2\pi R w,
$$
for some $w\in\Z$ which we identify with the winding number of the string. Most of the structure of the free scalar field remains unchanged, except for the zero-modes.
In this case, it proves convenient to write $X$ as a sum of holomorphic and anti-holomorphic functions:
$$
X(z,\bar z) = X_L(z) + X_R(\bar z),
$$
where
\ban
X_L(z) &= x_L -\I p_L \ln z +\I \sum_{n\neq0}\frac{\a_n}{nz^n}, \qquad X_R(\bar z) = x_R -\I p_R \ln \bar z +\I \sum_{n\neq0}\frac{\tilde\a_n}{n\bar z^n}.
\ean
In order to obtain the desired behaviour, we should now identify
\ban
\a_0 &= p_L = \frac{n}{R} + \frac{wR}{2}, \qquad \tilde\a_0 = p_R = \frac{n}{R} - \frac{wR}{2}.
\ean
The two position operators, $x_L$ and $x_R$ are not independent, since the physical center of mass is given by their sum: $x = x_L + x_R$.
Nevertheless, it is useful to write them as distinct to simplify the zero-mode algebra:
$$
[x_L,p_L] = [x_R,p_R] =\I .
$$
The momentum eigenstates now carry two labels, $\ket{k_L,k_R;0,0}$, and are associated with the primary operators  $\exp(\I k_L X_L(z) + \I k_R X_R(\bar z))$.
The energy of a typical state is therefore
$$
H\ket{k_L,k_R;N,\tilde N} = \left(\left(\frac{n}{R}\right)^2 + \left(\frac{wR}{2}\right)^2 +\sum_{m>0}m(N_m+\tilde N_m)\right)\ket{k_L,k_R;N,\tilde N}.
$$
As expected, this spectrum is now discrete. Furthermore, we see that the spectrum is invariant under the simultaneous interchange
$$
R\leftrightarrow\frac{2}{R},\quad n\leftrightarrow w,
$$
which we recognize as the symmetry of T-duality, discussed in Sect.~\ref{ss:Tdual}.\footnote{Recall that we are now working in units with $\a'=2$, and in general T-duality acts by $R\leftrightarrow \a'/R$.}

\subsubsection*{Free Fermion Field}
Aside from scalar fields, such as $X(z,\bar z)$, the other main building block for superconformal field theories are {\it fermion} fields, $\Psi(z,\bar z)$.
Fermion fields are somewhat peculiar because they anticommute among themselves: $\Psi_1\Psi_2 = -\Psi_2\Psi_1$, and this is the crucial feature that underlies the well-known ``Pauli Exclusion Principle".
In $d$ (Euclidean) dimensions, fermions transform in non-trivial spinor representations of $Spin(d)$. A Dirac spinor has $2^{\lfloor d/2 \rfloor}$ complex components, but this is often reducible.
For example, in $d=2$ a Dirac spinor takes the form
$$
\Psi(z,\bar z) = \begin{pmatrix} \psi(z) \\ \tilde\psi(\bar z) \end{pmatrix},
$$
and the components, $\psi(z)$ and $\tilde\psi(\bar z)$, transform separately as distinct Weyl spinors. In dimensions $d=2~{\rm mod}~8$ (and in particular $d=2$) Weyl spinors can be further reduced by imposing a Majorana (\ie reality) reality condition,
$\psi^\dagger(z) = \psi(z)$. This reduces $\psi$ to a single, real, fermionic component, which will be the focus of this example.
Much of this discussion parallels examination of the scalar field above, so we will only highlight the key differences.

We begin with the action for a single Majorana-Weyl fermion $\psi$:
$$
S = \frac{1}{4\pi} \int \D^2z\ \psi\bar\del\psi.
$$
The equation of motion that follows from this action is:
$$
\bar\del\psi(z)=0,
$$
so $\psi$ is holomorphic, and we can include a separate $\tilde \psi$ sector. Once again, the equation of motion holds up to insertions of $\psi$ at the same point, which leads to the OPE
$$
\psi(z) \psi(w) \sim \frac{1}{z-w}.
$$
Notice that this is antisymmetric under interchange of $z$ and $w$, consistent with the fermionic nature of the field. The energy-momentum tensor turns out to be
$$
T(z) = -\hlf \psi(z)\del\psi(z),
$$
and it is not hard to verify that it satisfies the OPE~\ref{eqn:TT}\ with $(c,\tilde c)=(\hlf,0)$. Some care is required in this computation, because we should only contract adjacent operators and permuting fermions will introduce important signs.
We can similarly compute:
$$
T(z) \psi(w) \sim \hlf\frac{\psi(z)}{(z-w)^2} + \hlf \frac{\del\psi(z)}{z-w} \sim \hlf \frac{\psi(w)}{(z-w)^2} + \frac{\del\psi(w)}{z-w},
$$
so we see that $\psi(z)$ is a primary with $(h,\tilde h)=(\hlf,0)$.

To introduce mode operators, we must decide on boundary conditions for the fermionic fields. As real-valued objects, there are only two natural possibilities:
periodic or antiperiodic, which are typically referred to as {\it Neveu-Schwarz (NS)}  and  {\it Ramond (R)} boundary conditions, respectively. Taking into account the conventional shift by $h$ in the Laurent expansion, we have
$$
\psi(z) = \sum_{r\in\Z+\n} \frac{\psi_r}{z^{r+1/2}}
$$
where $\n=1/2$ for $NS$ boundary conditions and $\n=0$ for $R$ boundary conditions. From the OPE, we can determine the algebra of the mode operators:
$$
\{\psi_r,\psi_s\}= \psi_r\psi_s + \psi_s\psi_r = \dd_{r,-s}.
$$
Note the use of the {\it anti-commutator}, as befits fermionic operators.
As usual, raising operators have $r<0$ and lowering operators have $r>0$.
In the $R$ sector, there is a zero-mode with $\{\psi_0,\psi_0\}=1$ which isomorphic to the Clifford algebra in one-dimension.\footnote{More generally, for a multi-component fermion field $\{\psi_0^i,\psi_0^j\} = \dd^{ij}$.} We will come back to this point momentarily.
From the mode expansion for $\psi(z)$, we can work out
$$
L_m = \frac14\sum_{r\in\Z+\n} (2r-m)\psi_{m-r}\psi_{r},
$$
where once again we should be careful about ordering in $L_0$, which we define as\footnote{Once again we are being sloppy about additive constants, which shift the zero point of the spectrum.}
\ban
 L_0 := \sum_{r\in\mathbb{N}+\n} r\psi_{-r}\psi_{r}.
\ean

With the exception of the zero-mode $\psi_0$ in the $R$ sector, every mode operator squares to zero. Therefore, the occupation numbers are binary $N_r\in\{0,1\}$.
In the $NS$ sector, a typical state is then of the form
$$
\ket{N}_{NS} = \ket{N_{1/2},N_{3/2},\ldots}_{NS} = \ldots\psi_{-3/2}^{N_{3/2}}\psi_{-1/2}^{N_{1/2}}\ket{0}_{NS},
$$
with energy given by
$$
L_0\ket{N}_{NS} = \sum_{r\in\mathbb{N}+1/2} rN_r \ket{N}_{NS}.
$$
In the $R$ sector, if $\ket{0}$ is a vacuum state annihilated by $\psi_n$ for $n>0$, then $\psi_0\ket{0}$ defines another vacuum state, since $\{\psi_n,\psi_0\}=0$ $\forall n\neq0$.
Thus, the $R$ vacuum is actually two-fold degenerate, which we will label as
$$
\ket{-}_R,\quad \ket{+}_R = \sqrt{2}\psi_0 \ket{-}_R.
$$
Since $(\sqrt{2}\psi_0)^2 = 1$, the degeneracy is only two-fold. Thus, a typical state in the $R$ sector is given by
$$
\ket{\pm;N}_R = \ket{\pm;N_{1},N_{2},\ldots}_{R} = \ldots\psi_{-2}^{N_{2}}\psi_{-1}^{N_{1}}\ket{\pm}_{R}
$$
with energies
$$
L_0\ket{\pm;N}_R = \sum_{n\in\mathbb{N}} nN_n\ket{\pm}_R.
$$
Fermion fields, with their distinct $NS$ and $R$ sectors, will play an important role in the superconformal theories, to which we can finally turn our attention.

\section{$N=(2,2)$ Superconformal Field Theories}\label{sec:SCFT}
After the whirlwind review of conformal field theories in the last chapter, we can now focus on the $N=(2,2)$ superconformal field theories (SCFTs), which are relevant for Calabi-Yau manifolds and mirror symmetry.
Section~\ref{ss:super}\ introduces the concepts of supersymmetry, superconformal algebras and their representations.
An important subset of states in an $N=(2,2)$ SCFT are the chiral primaries, which we study in Sect.~\ref{ss:chiral}\ along with some of their essential properties.
Another important feature in $(2,2)$ theories is called spectral flow, discussed in Sect.~\ref{ss:spectral}, which is ultimately the reason these models will exhibit mirror symmetry.
Section~\ref{ss:geometry}\ will develop a remarkable relation between the chiral primary states in an $N=(2,2)$ SCFT on the one hand, and the cohomology rings of Calabi-Yau manifolds on the other.
Mirror symmetry will emerge as a trivial automorphism of the SCFT, though the geometric implications are far from trivial.
We will present several classes of examples in Sect.~\ref{ss:examples2}\ to illustrate this structure.
We close with brief comments on some of the physical applications of mirror symmetry in Sect.~\ref{ss:summary2}.
Many of these ideas were first formulated in~\cite{Lerche:1989uy}, and~\cite{Greene:1996cy}\ provides an excellent review of this material as well as more advanced topics we do not have time to cover.

\subsection{Superconformal Groups}\label{ss:super}
Since superconformal symmetries will play a central role in the remainder of these notes, let us take a moment to understand their general structure.
Roughly speaking, {\it supersymmetry} is an extension of the Poincar\'e group, generated by $P_\m$ and $J_{\m\n}$, by anticommuting ``supercharges" $Q_\a$, such that
$$
\{Q_\a, Q_\beta\}= \g^\m_{\a\b}P_\m,
$$
where $\{A,B\}=AB+BA$ is the anticommutator of two operators, and $\g^\m$ generate the Clifford algebra $\{\g^\m,\g^\n\} = 2\eta^{\m\n}$.\footnote{Additionally, supersymmetry groups contain $R$-symmetries as ``internal" sub-groups, which act non-trivially on the supercharges only.}
Thus, $Q$ behaves as a sort of ``square-root" of translations. It is possible to have several such generators, $Q^A_\a$ with $A=1,2,\ldots,N$, which we refer to as {\it $N$-extended supersymmetry}.
Similarly, a superconformal group is an enhancement of the conformal group into a ``supergroup", \ie\ a Lie group generated by a $\Z_2$-graded algebra, by including generators $Q_\a$ and $S_\a$. In analogy with the $Q_\a$, we can think of $S_\a$ as a kind of ``square-root" of the SCTs $K_\m$.
For our purposes, the following rough definition will suffice:
\begin{definition}
  In $d>2$, the {\it $N$-extended superconformal group of $(\R^d,\eta_{\m\n})$} is the most general $\Z_2$-graded extension the (even) conformal group $SO(2,d)$, by the odd generators $Q^A_\a$ and $S^A_\a$, where $A=1,2,\ldots N$, transforming in spinor representations of $Spin(1,d-1)$.
  The even elements of the group are called {\it bosons} and the odd elements are called {\it fermions}.
\end{definition}
We will not dwell here on features such as uniqueness of these extensions or the detailed form of the complete super-algebras, which may be found in \eg~\cite{West:1990tg}.
Instead we will focus on the case of interest, $d=2$ with Euclidean signature, where once again special considerations are required because of the infinite dimensional nature of the conformal group.
For the most part, the concepts developed in Sect.~\ref{sec:conf}\ lift to the supersymmetric setting rather straightforwardly.
Therefore, we will proceed rather quickly, pausing only to emphasize the new features that arise.

\subsubsection*{$N=1$ Super-Virasoro Algebras}
There exist two minimal (\ie\ $N=1$) supersymmetric extensions of the Virasoro algebra, corresponding to the two classes on boundary conditions for the fermions in the theory: the {\it Neveu-Schwarz (NS) algebra} and the {\it Ramond (R) algebra}.
In addition to the Virasoro generators, $L_n$, we introduce odd generators $G_r$.
\begin{definition}
  The  $N=1$ super-Virasoro algebras (with central charge $c$), are generated by mode operators $L_n$ (even) and $G_r$ (odd) subject to the graded commutation relations:
  \ban
  [L_m,L_n] &= (m-n)L_{m+n} +\frac{c}{12}m(m^2-1)\dd_{m,-n} \\
  [L_m,G_r] &= \left(\frac{m}{2}-r\right)G_{m+r} \\
  \{G_r, G_s\} &= 2L_{r+s} +\frac{c}{12}(4r^2-1)\dd_{r,-s}.
  \ean
For the {\it Neveu-Schwarz algebra} $r\in\Z+\hlf$, while for the {\it Ramond algebra} $r\in\Z$.
\end{definition}
Notice that in the $NS$ sector, the operators $L_0, L_{\pm1}, G_{\pm\hlf}$ satisfy a closed sub-algebra independent of the central charge. This is holomorphic portion of the global superconformal algebra. No such sub-algebra exists in the $R$ sector.
\begin{definition}
  The {\it supercurrent} $G(z)$ is the local operator with Laurent coefficients $G_r$:
  $$
  G(z) = \sum_{r\in\Z+\n} \frac{G_r}{z^{r+3/2}},
  $$
  where $\nu=\hlf$ (resp. $0$) for the $NS$ (resp. $R$) algebras.
\end{definition}
The shift in the powers of $z$ suggests that $G(z)$ has weight $h=\frac{3}{2}$. This suspicion is indeed confirmed by considering the algebra in OPE form.
\begin{proposition}
By the usual contour argument, the $N=1$ super-Virasoro algebras are equivalent to the following OPEs:
\ban
T(z)T(w) &\sim \frac{c/2}{(z-w)^4} +\frac{2T(w)}{(z-w)^2} + \frac{\del T(w)}{z-w}\\
T(z)G(w) &\sim \frac{(3/2)\,G(w)}{(z-w)^2}  + \frac{\del G(w)}{z-w}\\
G(z)G(w) &\sim \frac{2c/3}{(z-w)^3} + \frac{2T(z)}{z-w}.
\ean
\end{proposition}
Rather than develop the representation theory of the $N=1$ algebra here, we will jump ahead to the case of $N=2$ (extended) superconformal algebra.

\subsubsection*{The $N=2$ Superconformal Algebra}
As the name suggests, the $N=2$ algebra contains two fermionic supercurrents, $G^i(z)$ for $i=1,2$. In addition, the $N=2$ algebras includes an $SO(2)$ current, $J(z)$,\footnote{This is an example of an $R$-symmetry, alluded to in an earlier footnote.} which together with $T(z)$ completes the list of generators.
The supercurrents $G^i(z)$ transform as a doublet under the $SO(2)$ symmetry, so it proves convenient to combine these into the complex combinations
$$
G^{\pm}(z) = \frac{1}{\sqrt{2}}\left(G^1(z) \pm \I G^2(z)\right),
$$
which carry opposite charges under the $U(1)\simeq SO(2)$ symmetry. Note that complex conjugation is therefore equivalent to charge inversion.
\begin{definition}
The {\it $N=2$ super-Virasoro algebra} is generated by the local operators $T(z), G^{\pm}(z),$ and $J(z)$, with weights $\{2,\frac{3}{2},\frac{3}{2},1\}$ and OPEs:
\ban
T(z)T(w) &\sim \frac{c/2}{(z-w)^4} +\frac{2T(w)}{(z-w)^2} + \frac{\del T(w)}{z-w} \\
T(z)G^{\pm}(w) &\sim \frac{(3/2)\, G^\pm(w)}{(z-w)^2}  + \frac{\del G^\pm(w)}{z-w} \\
G^\pm(z)G^\mp(w) &\sim \frac{2c/3}{(z-w)^3} \pm\frac {2J(w)}{(z-w)^2}+ \frac{2T(z)\pm \del J(w)}{z-w} \\
G^\pm(z)G^\pm(w) &\sim 0\\
T(z)J(w) &\sim \frac{J(w)}{(z-w)^2} +\frac{\del J(w)}{z-w} \\
J(z)G^\pm(w) &\sim \pm \frac{G^\pm(w)}{z-w} \\
J(z)J(w) &\sim \frac{c/3}{(z-w)^2}.
\ean
\end{definition}
\begin{proposition}\label{prop:N=2}
  By making the mode expansions
  $$
  T(z) = \sum_{n\in\Z}\frac{L_n}{z^{n+2}},\quad G^\pm(z) = \sum_{r\in\Z\pm\n}\frac{G^\pm_r}{z^{r+3/2}},\quad J(z) = \sum_{n\in\Z}\frac{J_n}{z^{n+1}},
  $$
  the OPEs of the $N=2$ super-Virasoro algebra are equivalent to the following graded commutation relations:
\ban
[L_m,L_n] &= (m-n)L_{m+n} +\frac{c}{12}m(m^2-1)\dd_{m,-n} \\
[L_m,G^\pm_{r}] &= \left(\frac{m}{2}-r\right)G^\pm_{m+r} \\
\{G^+_{r}, G^-_{s}\} &= 2L_{r+s} +(r-s)J_{r+s} +\frac{c}{12}(4r^2-1)\dd_{r,-s} \\
\{G^+_r,G^+_s\} &= \{G^-_r,G^-_s\} =0\\
[L_m,J_n] &= -nJ_{m+n} \\
[J_m,G^\pm_{r}] &= \pm G^\pm_{m+r} \\
[J_m,J_n] &= \frac{c}{3}m\, \dd_{m,-n}.
\ean
\end{proposition}
As usual, the proof simply follows by applying the standard contour argument. A new feature of the $N=2$ algebra is that $\n$ is now free to take any real value, though clearly the algebras labeled by $\n$ and $\n+1$ are isomorphic.
Thus, there appears to be one parameter family of $N=2$ algebras, labeled by $\n\in[0,1)$. The cases of physical interest remain $\n=0,\hlf$, which we continue to refer to as the Ramond and Neveu-Schwarz sectors.
However, we will see in Sect.~\ref{ss:spectral}\ that all of these $N=2$ algebras are in fact isomorphic.\footnote{This is the reason we have been referring to the $N=2$ algebra in the singular, unlike the $N=1$ cases.} This simple fact will lie at the very heart of mirror symmetry.

\subsubsection*{Representations of the $N=2$ Algebra}
Exactly as we did for the (standard) Virasoro algebra, we must divide the mode operators of the $N=2$ algebra into raising and lowering types as dictated by the algebra in Prop.~\ref{prop:N=2}.
Since $J_0$ commutes with $L_0$, we must also label states by their eigenvalue of this conserved charge.\footnote{In the Ramond sector, when $\n=0$, we must also deal with $G^\pm_0$, which also commute with $L_0$. We define {\it Ramond groundstates} to be those annihilated by both $G_0^\pm$.}
As before, the irreducible representations of the $N=2$ algebra come in (super-)conformal families built on lowest weight states.
\begin{definition}
In an $N=2$ superconformal field theory, a state $\ket{\cO}$ is called a {\it primary of weight $h$ and charge $q$}, if
\ban
&L_0\ket{\cO} = h\ket{\cO},\qquad J_0\ket{\cO} = q\ket{\cO}\\
L_n\ket{\cO} &= G_r^\pm\ket{\cO} = J_m\ket{\cO} =0,\quad \forall~n,r,m>0.
\ean
Equivalently, by the state-operator correspondence, the local operator $\cO(z)$ is called a {\it primary of weight $h$ and charge $q$}, if
\ban
T(z)\cO(w) &\sim \frac{h}{(z-w)^2} + \frac{\del\cO(w)}{z-w} \\
J(z)\cO(w) &\sim \frac{q\,\cO(w)}{z-w} \\
G^\pm(z)\cO(w) &\sim \frac{\tilde \cO^\pm(w)}{z-w}
\ean
where $\tilde\cO^\pm$ are the {\it superpartners} of $\cO$.
Finally, an $N=2$ superconformal field theory is {\it unitary} if the inner product $\braket{\ }{\ }$ is positive and
$$
L_n^\dag = L_{-n},\quad \left(G_r^\pm\right)^\dag = G_{-r}^\mp,\quad  J_n^\dag = J_{-n}.
$$
\end{definition}

\subsection{Chiral Rings}\label{ss:chiral}
For now, let us restrict our attention to the $NS$ sector, with $\n=\hlf$. Then, there exists a very special subset of primary operators in an $N=2$ theory, which will be our focus for the remainder of these notes.
\begin{definition}
$\cO(z)$ is called a {\it chiral} if $G^+_{-1/2}\ket{\cO}=0,$ and {\it anti-chiral} if $G^-_{-1/2}\ket{\cO}=0.$
An operator that is both primary and (anti-)chiral called an {\it (anti-)chiral primary}.
\end{definition}
If $\cO(z)$ is a chiral or anti-chiral primary, then this implies
$$
G^+(z) \cO(w)\sim 0,\quad {\rm or}\quad G^-(z) \cO(w)\sim 0.
$$
Thus, half of the superpartner $\tilde \cO^\pm$ are absent for these special classes of operators. In physics parlance, the superconformal families associated with (anti-)chiral primaries are in {\it short multiplets},\footnote{The terminology {\it BPS multiplets} is also used in this case.} because they contain fewer states than a typical irreducible representation.

Taking into account the anti-holomorphic operators $\tilde G^{\pm}(\bar z)$, since we are really interested in $N=(2,2)$ superconformal theories, we have four distinguished subsets of primary operators, which we can label in an obvious manner by:  $(c,c),~(a,c),~(c,a),~(a,a)$.
However, not all of these sectors are independent since charge conjugation takes CPOs into APOs. Thus in an $N=(2,2)$ superconformal theory, there are two distinguished sectors of chiral/anti-chiral primary operators, which we choose to be
$$(c,c)\simeq(a,a)^*,\quad {\rm and}\quad (a,c)\simeq(c,a)^*.$$
It is possible to perform a ``topological twist" of an $N=(2,2)$ theory~\cite{Witten:1988xj,Witten:1989ig,Witten:1991zz,Vafa:1991uz}, with the results that the entire spectrum of the theory is truncated only to one of these two sectors.
These are the so-called {\it A-model} and {\it B-model} of topological string/field theory, which are likely more familiar to this audience than the full-blown $(2,2)$ theory that we have been studying.
While these specializations have played (and continue to play) a central role in the development of mirror symmetry, particularly for making precise mathematical statements, we will not discuss them any further here.
However, the results that are about to follow should clarify why these particular sets of operators garner so much attention from mathematicians and physicists alike.

\begin{theorem}\label{thm:ring}
In an $N=2$ superconformal theory, the set of chiral primaries form a non-singular and closed ring, ${\cal R}_c$, under the operation of operator product {\rm at the same point}. Moreover, if the theory is unitary and non-degenerate, then $\cR_c$ is finite.
\end{theorem}
Similarly there are (finite) rings $\cR_a$ of anti-chiral primaries, and the tensor products $\cR_{cc}=\cR_c\otimes\tilde\cR_c$ and $\cR_{ac}=\cR_a\otimes\tilde\cR_c$. Before we can prove this important result, we need another important property of (anti-)chiral primaries.
\begin{lemma}\label{lem:BPS}{\rm (BPS bound)}
In a unitary $N=2$ SCFT, the weights and charges of every local operator $\cO$ must obey the inequality
$$ h\geq\frac{|q|}{2}.$$
This bound is saturated by $(i)$ $h=+q/2$ iff $\cO$ is a chiral primary, or $(ii)$ $h=-q/2$ iff $\cO$ is an anti-chiral primary.
\end{lemma}
\begin{proof}
\smartqed
Recall that in a unitary $N=2$ theory the norm on the Hilbert space of states is positive and: $$\left(G_r^\pm\right)^\dag = G^\mp_{-r}.$$
Therefore, for any state $\ket{\cO}$ we have
\ban
0 &\leq ||G^\pm_{-1/2}\ket{\cO}||^2 + ||G^\mp_{1/2}\ket{\cO}||^2 = \bra{\cO}\{G^\mp_{1/2},G^\pm_{-1/2}\}\ket{\cO} = \bra{\cO}\left(2L_0\mp J_0\right)\ket{\cO} = 2h \mp q,
\ean
which proves the inequality.

To demonstrate the second statement, we show the reverse direction first. Assume that $\cO$ is a chiral primary. Then, by definition
$$ G^\pm_{1/2}\ket{\cO} = G^+_{-1/2}\ket{\cO} =0,$$
and it follows that $h=q/2$. Similarly, if $\cO$ is an anti-chiral primary, then it follows immediately that $h=-q/2$. Going in the other direction, we now assume that $\cO$ is an operator such that $h=\pm q/2$, which tells us
$$ G^\pm_{-1/2}\ket{\cO} = G^\mp_{1/2}\ket{\cO} =0.$$
The first condition tells us that $\cO$ is {\it (anti-)chiral}, but not necessarily primary.\footnote{The second condition is certainly required to be primary, but is not sufficient.} Thus, it remains to show that:
$$ L_n\ket{\cO} = G_r^\pm\ket{\cO} =G^\mp_{s+1}\ket{\cO}= J_m\ket{\cO} =0,\quad \forall~n,r,s,m>0.$$
However, all of these operators reduce the value of $h\mp q/2$:
\ban
[(L_0\mp\tfrac12 J_0),L_n] = -nL_n,\quad &\quad  [(L_0\mp\tfrac12 J_0),G^\pm_r] = -(r+\tfrac12)G^\pm_r ,\\
[(L_0\mp\tfrac12 J_0),J_m] = -mJ_m , \quad &\quad [(L_0\mp\tfrac12 J_0),G^\mp_{s+1}] = -(s+\tfrac12)G^\mp_{s+1}.
\ean
Since $\ket{\cO}$ already saturates the BPS bound, $h= |q|/2$, these new states must all vanish. Thus, $\cO$ is primary as well.
\qed
\end{proof}
Armed with these important facts, we can now return to proving the ring structure of the chiral primaries.
\begin{proof}[of Thm.~\ref{thm:ring}]
\smartqed
Consider the OPE of two chiral primary operators $\cO_1$ and $\cO_2$:
$$
\cO_1(z_1)\cO_2(z_2) = \sum_i (z_1-z_2)^{h_i-h_1-h_2}\cO_i(z_2),
$$
where we have absorbed the structure coefficients $c_{12i}$ into the operators $\cO_i$. Charge conservation, together with the BPS bound~\C{lem:BPS}, gives us:
$$
h_i \geq \frac{q_i}{2} = \frac{q_1+q_2}{2} = h_1 + h_2.
$$
Therefore, the OPE of two chiral primaries is non-singular, and we are free to take the limit $z_1\rightarrow z_2$. However, the only terms that survive that limit must have
$$ h_i = h_1 + h_2 = \frac{q_i}{2},$$
so they must also be chiral primaries, and we have
$$
\cO_1(z)\cO_2(z) = \sum_{i\in \{CPO\}} \cO_i(z),
$$
where $\{CPO\}$ denotes the set of chiral primaries operators. Thus, the chiral primaries form a closed ring.

To show finiteness of this ring, we go back to the $N=2$ algebra and consider
$$
\{G^-_{3/2},G^+_{-3/2}\} = 2L_0-3J_0 +\frac{2}{3}c.
$$
In a unitary theory, $G^-_{3/2}$ and $G^+_{-3/2}$ are each others adjoints, and so sandwiching their bracket by any chiral primary states $\ket{\cO}$ gives us:
$$
0\leq \bra{\cO}\{G^-_{3/2},G^+_{-3/2}\}\ket{\cO} = \bra{\cO}(2h-3q+2c/3 )\ket{\cO} = -4h +2c/3.
$$
This sets an upper-bound on the weight of a chiral primary: $h\leq c/6$. Recall, from Prop.~\ref{prop:unitary}, that in a unitary CFT $h\geq0$ for all operators. Thus, the allowed weights of chiral primaries must lie in the range
$$
0\leq h\leq \frac{c}{6},
$$
and so in a non-degenerate theory (where the values of $h$ are discrete), there are only finitely many possibilities.
\qed
\end{proof}
While the existence of a non-singular ring structure is an interesting feature of the theory, and useful for computations, this fact alone is not terribly exciting.
What makes the chiral rings interesting is their relationship to algebraic geometry. This next theorem, and especially its proof, will begin to uncover this connection.
\begin{theorem}\label{thm:coh}
The chiral ring $\cR_c$ is isomorphic to the cohomology ring of $G^+_{-1/2}$.
\end{theorem}
\begin{proof}
\smartqed
First, we note that $$\left(G^+_{-1/2}\right)^2=0,$$ so its cohomology is a well-defined notion. Clearly, $\cR_c \subseteq{\rm Ker}~ G^+_{-1/2}$,
so it remains to show that every chiral operator can be written as a chiral primary modulo ${\rm Im}(G^+_{-1/2})$.
This can be done by an analog of the Hodge decomposition for forms, but we will follow a different route instead.
The idea will be to demonstrate an equivalence between the chiral structures of the $N=2$ algebra and the $\bar\del$-cohomology of complex manifolds, where the equivalent result is well-known.
In the process, we uncover a deep relationship between the $N=2$ chiral rings and complex geometry, which will be expanded upon in the sequel.

The map between these two structures is given by the following dictionary:
\ban
\bar\del\ &\leftrightarrow\ G^+_{-1/2}\\
\bar\del^\dag\ &\leftrightarrow\ G^-_{1/2}\\
\Delta\ &\leftrightarrow\  2L_0 -J_0 \\
{\rm deg}\ &\leftrightarrow\ J_0.
\ean
Crucially, this identification respects the correct algebras on both sides, for instance:
$$\Delta = \bar\del\bar\del^\dag +\bar\del^\dag\bar\del\ \leftrightarrow\ \{G^+_{-1/2},G^-_{1/2}\} = 2L_0-J_0.$$
Once the operators are identified, the rest of the map falls into place:
\ban
{\rm chiral}\ &\leftrightarrow\ \bar\del{\rm -closed}\\
{\rm chiral~primary}\ &\leftrightarrow\ {\rm harmonic}.
\ean
Thus, the statement of the theorem is just that every (de Rham or Dolbeault) cohomology class contains a harmonic representative, which is a well-known fact.
\qed
\end{proof}
At this point, by using the above dictionary we could translate many of the well-known results from Dolbeault cohomology into corresponding statements about the $N=2$ chiral rings.
We will content ourselves for the moment with just one.
\begin{corollary}{\rm (Hodge decomposition)} In an $N=2$ theory, any state $\ket{\cO}$ can be written in the form
$$
\ket{\cO} = \ket{\cO_0} + G^+_{-1/2}\ket{\cO_1} + G^-_{1/2}\ket{\cO_2},
$$
where $\cO_0$ is a chiral primary, and $\cO_1$ and $\cO_2$ are some other operators. In particular, if $\cO$ is chiral then $\cO_2=0$.
\end{corollary}
This decomposition would have been required in proving Thm.~\ref{thm:coh}, but by cleverly mapping the structure into a familiar language we were able effectively side-step the issue.
We will elaborate further on the analogy between chiral rings and complex geometry in the remainder of these notes. Under some simple, and fairly obvious, assumptions we will
see that the rings $\cR_{cc}$ and $\cR_{ac}$ can be interpreted as the Dolbeault cohomology rings $H^{*,*}(X)=\oplus_{p,q}H^{p,q}(X)$ on some K\"ahler manifold $X$.
Before doing so, first must discuss a couple other important features of $N=2$ theories.

\subsubsection*{Marginal Operators}
Recall the example of the compactified boson from Sect.~\ref{ss:examples}, which\footnote{For simplicity, we suppress the normalization constant $1/4\pi$ in front of the action.} is specified by the action
$$
S_0 = \int \D^2z\ \del X\bar\del X,
$$
with $X\sim X+2\pi R$. Geometrically, this theory describes the mapping of a two dimensional worldsheet into a space containing a circle of radius $R$.
For the purposes of this discussion, it will be more convenient to work with the rescaled field $\tilde X(z,\bar z):=R X(z,\bar z)$ so that
$$
S_0 =R^2 \int \D^2z\  \del \tilde X \bar \del \tilde X,
$$
and now $\tilde X\sim \tilde X +2\pi$. Thus $\tilde X$ is independent of the modulus of the circle, $R$, which now appears as a parameter in the action.
Suppose now we wish to deform the size of the circle by an infinitesimal amount: $R\rightarrow R+\varepsilon$. It follows that the action $S_0$ is also deformed:
$$
S_0\rightarrow S_\ve = S_0 + 2\ve R\int \D^2z\ \del\tilde X \bar\del \tilde X +O(\ve^2).
$$
The compactified boson therefore defines a family of CFTs labeled by the modulus $R$.
Theories corresponding to different values of $R$ are therefore connected by adding the operator $\cO = \del\tilde X\bar\del \tilde X$ to the original action $S_0$.
Notice that this deforming operator has $(h,\tilde h)=(1,1)$, and this is exactly compensated for by the measure $\D^2z=\D z \D\bar z$, so that the new integrated action is also conformally invariant.
In this simple example, it is clear that both $S_0$ and $S_\ve$ define conformally invariant theories since we are only changing the radius of a circle. More generally, such a $(1,1)$ deformation may break the conformal invariance of the theory once it used to deform the action.
\begin{definition}
  A local operator $\cO(z,\bar z)$ is called {\it marginal} if it has weights $(h,\tilde h)=(1,1)$, and can therefore be used to deform a CFT.
  A marginal operator is called {\it truly marginal} if it remains $(1,1)$ after being added to the action of a conformal theory.
\end{definition}
We now wish to study the truly marginal operators of an $N=(2,2)$ theory, since these will lead us to families of $(2,2)$ SCFTs.
In particular, we would like to examine the link between truly marginal operators and chiral primary operators.
In an $N=(2,2)$ theory, in addition to being conformally invariant a marginal deformation must also respect the $U(1)\times U(1)$ symmetry generated by $J(z)$ and $\tilde J(\bar z)$.
So, clearly, the chiral primaries themselves cannot be marginal, since the only neutral chiral primary is the vacuum state with $h=\tilde h=0$.
However, it is possible for certain superpartners of chiral primaries to be (truly) marginal.
\begin{proposition}
    Let $\cO_{(\pm1,+1)}$ be elements of the rings $\cR_{cc}$ and $\cR_{ac}$ with charges $ q=\pm1$, and $\tilde q=+1$. Then the operators
  $$
  \hat\cO_{(\pm1,+1)}(w,\bar w) := \oint_w \frac{\D z}{2\pi\I} \oint_{\bar w}\frac{\D\bar z}{2\pi\I}\ G^\mp(z)\tilde G^-(\bar z) \cO_{(\pm1,+1)}(w,\bar w),
  $$
  are truly marginal. Furthermore, every truly marginal operator in an $N=(2,2)$ theory can be associated with an element of the $\cR_{cc}$ or $\cR_{ac}$ ring in this manner.
\end{proposition}
\begin{proof}
\smartqed
  Here we will only provide a sketch of the proof. The essential point is that, in both cases, it is clear that $\hat\cO_{(\pm1,1)}$ have the correct weights and charges to be marginal, namely $h=\tilde h=1$ and $q=\tilde q=0$.
  This follows from the basic fact that $G^\pm(z)$ has $(h,\tilde h)=(\hlf,0)$ and $(q,\tilde q)=(\pm1,0)$, and similarly for $\tilde G^\pm(\bar z)$.
  We leave the proof that these define {\it truly} marginal deformations to the references, though we will see in examples that this fact often correlated with the unobstructedness of geometric moduli, which is well-understood.
  That every (truly) marginal operator can be written in this way follows by applying the dictionary used in Thm.~\ref{thm:coh}\ to derive an analog of the $\del\bar\del$-lemma.
\qed
\end{proof}
Thus, among the elements of the $\cR_{cc}$ and $\cR_{ac}$ rings, those with $h=\tilde h=\hlf$ will play a special role, since they will parameterize the connected families of $N=(2,2)$ superconformal theories.

\subsection{Spectral Flow}\label{ss:spectral}

Recall that there exists a parameter $\n$ in the $N=2$ algebra, which determines the periodicity of $G^\pm(z)$ under $z\rightarrow e^{2\pi\I}z$. In particular, since
$$
G^\pm(z) = \sum_{r\in\Z\pm\n}\frac{G^\pm_r}{z^{r+3/2}}, \quad \Rightarrow \quad G^\pm(\E^{2\pi\I}z) = -\E^{\mp2\pi\I\n}G^\pm(z).
$$
The $NS$ and $R$ sectors of the theory correspond to $\n=1/2$ and $\n=0$, so that $G^\pm$ are, respectively, periodic and antiperiodic.
However, as we alluded to earlier, these sectors are not distinct.
\begin{proposition}
  The one-parameter family of $N=2$ algebras, labeled by $\n$, are all isomorphic.
\end{proposition}
\begin{proof}
\smartqed
The basic idea is to construct a one-parameter family of operators that interpolate between the different values of $\n$. To that end, consider the following ``twisted" operators:
  \ban
  L_n^\th &:= L_n +\th J_n +\frac{c}{6}\th^2 \dd_{n,0}\\
  J_n^\th &:= J_n +\frac{c}{3}\th\dd_{n,0}\\
  G^{\th\pm}_r &:=  G^\pm_{r\pm\th}.
  \ean
If the original operators (for $\th=0$) are defined for a given value of $\n$, then their twisted versions (for $\th\neq0$) are in the sector $\n+\th$.
So, the twisted operators are simply linear combinations of the untwisted operators at a shifted value of $\n$.
In particular, the algebra of the twisted operators is isomorphic to the algebra of the untwisted operators in the sector $\n+\th$.
If we can show that the algebra of the twisted operators are equivalent for all values of $\th$, then we will have demonstrated the isomorphism of the $N=2$ algebra for all values of $\n$.

Indeed the twisted operators satisfy the same algebra for all values of $\th$, and moreover that algebra is precisely the $N=2$ algebra (for any fixed value of $\n$). The proof of this claim is by brute force computation. For example:
$$
[L^\th_n,G^{\th\pm}_r] = [L_n, G^\pm_{r\pm\th}] +\th[J_n,G^\pm_{r\pm\th}] =\left(\frac n2 -r\right)G^{\pm}_{n+r\pm\th}= \left(\frac n2 -r\right)G^{\th\pm}_{n+r},
$$
and the other (graded) commutation relations can be worked out similarly, so we will not reproduce them here. The point is that the algebra of the twisted operators are all equivalent to the $N=2$ algebra, and so the different values of $\n$ all yield isomorphic algebras.
\qed
\end{proof}
Since different values of $\n$ yield isomorphic algebras, it follows that the corresponding representations isomorphic as well.
Thus, there must exist a unitary transformation which maps between the different sectors labeled by $\n$.
If we denote by $\cH_\nu$ the Hilbert space in the sector $\n$, then we have a unitary transformation
$$
\cU_\th:{\cal H}_\n \rightarrow {\cal H}_{\nu+\th},
$$
such that for any operator $\cO$ that acts on $\cH_\n$, there exists an operator $\cO_{\th}$ that acts on $\cH_{\n+\th}$ given by
$$
\cO_{\th} = \cU_\th \cO\cU_{\th}^{-1}.
$$
In particular, the action of $\cU_\th$ on the $N=2$ generators gives:
  \ban
\cU^{-1}_\th L_n\cU_\th &= L_n^\th  \\
\cU^{-1}_\th J_n\cU_\th &= J_n^\th \\
\cU^{-1}_\th G^\pm_r\cU_\th &= G^{\th\pm}_r ,
\ean
and similarly for the anti-holomorphic generators.
\begin{definition}
The map between different sectors, $\cH_\n\rightarrow\cH_{\n+\th}$, is called {\it spectral flow} by an amount $\th$, and $\cU_\th$  is called the {\it spectral flow operator}.
\end{definition}
Notice that spectral flow by $\th=\pm\hlf$ induces an isomorphism between the $NS$ and $R$ sectors of the theory,\footnote{This equivalence, together with an integral charge constraint, guarantees the existence of supersymmetry in the target space (as opposed to worldsheet) theory~\cite{Banks:1987cy}, since bosons originate in the $NS$ sector and fermions originate in the $R$ sector.}
while spectral flow by an integral amount maps the $NS$ and $R$ sectors back to themselves.
Some care must be taken when comparing the Hilbert spaces $\cH_\nu$ in different sectors. Regarded simply as a collections of states, then spectral flow provides an isomorphism between $\cH_\nu$ and $\cH_{\n+\th}$.
However, as modules of the $N=2$ algebra these Hilbert spaces are not isomorphic, since they are acted upon by different (though isomorphic) algebras.\footnote{A nice discussion of this subtle point can be found in~\cite{Greene:1996cy}.}
This is a consequence of the fact that the charges and weights of each state varies with $\th$.
\begin{proposition}
  Let $\cO_0$ be an operator of weight $h_0$ and charge $q_0$, and let $\cO_\th$ be its image under spectral flow by an amount $\th$ with weight and charge $h_\th$ and $q_\th$. Then
\ban
q_\th=q_0-\frac{c}{3}\th,\quad h_\th=h_0-\th q_0 +\frac c6 \th^2.
\ean
In particular, if $\cO_0\in\cR_c$ , then $\cO_{1/2}$ is an {\rm R} ground state, and $\cO_{1}\in\cR_a$.
\end{proposition}
\begin{proof}
\smartqed
Let $\ket{\th}$ denote the vacuum in the $\th$ twisted sector. Then states $\ket{\cO_\th}$ and $\ket{\cO_0}$ are related by spectral flow:
$$
\ket{\cO_\th} = \cO_\th\ket{\th} = \cU_\th \cO_0 \cU_\th^{-1}\ket{\th} = \cU_\th\cO_0\ket{0} = \cU_\th\ket{\cO_0}.
$$
  Then, on the one hand we have
  $$
  J^\th_0\ket{\cO_\th} = \cU_\th J_0 \cU_\th^{-1}\ket{\cO_\th} = \cU_\th J_0 \ket{\cO_0} = q_0\cU_\th\ket{\cO_0} = q_0\ket{\cO_\th},
  $$
  but on the other hand,
  $$
  J^\th_0\ket{\cO_\th} = \left(J_0+\frac c3 \th\right)\ket{\cO_\th} = \left(q_\th+\frac c3 \th\right)\ket{\cO_\th}.
  $$
  Putting these together, we conclude that
  $$
  q_\th = q_0-\frac c3 \th.
  $$
By similar reasoning, we can also conclude
$$
h_\th = h_0-\th q_\th -\frac c6 \th^2 = h_0 - \th q_0 +\frac c6 \th^2.
$$

Now, suppose that $\cO_0\in \cR_c$ so that $h_0=q_0/2$, and we are in an $NS$ sector with $\n\in\Z+\hlf$. Spectral flow by $\th=\hlf$ takes us to an  $R$ sector, for which
$$
h_{1/2} = h_0-\frac{q_0}{2} +\frac{c}{24} = \frac{c}{24}.
$$
We have not discussed the Ramond sector in much detail, but we can place a lower bound on the allowed weights much as we did in the $NS$ sector. Since $G^\pm_0$ are adjoints to one another, then for any state $\ket{\psi}_R$ in the Ramond sector,
$$
0\leq {}_R\bra{\psi}\{G^+_0,G^-_0\}\ket{\psi}_R = {}_R\bra{\psi}(2L_0 -c/12)\ket{\psi}_R = 2h-c/12.
$$
In particular a Ramond groundstate, which must be annihilated by $G^\pm_0$, saturates the lower bound of $h=c/24$. Thus, the chiral primary states flow to the Ramond sector groundstates.
If instead we flow by $\th=1$, we return to an $NS$ sector, but now
$$
q_1 = q_0 +\frac c3,
$$
and
$$
h_1 = h_0-q_0 +\frac c6 = -\frac{q_1}{2} + \frac c6 = -\frac{q_1}{2}.
$$
So, indeed, a chiral primary operator flows to an anti-chiral primary.
\qed
\end{proof}

\subsection{Calabi-Yaus and Mirror Symmetry}\label{ss:geometry}
We have seen that the spectral flow operator induces an isomorphism of the $\cR_c$ and $\cR_a$ chiral rings.
However, recall that our real interest is $N=(2,2)$ theories, so we must include the anti-holomorphic sector as well.
Let $\cU_{\th,\tilde\th}$ denote the combined spectral flow operator for the two sectors.
Then $\cU_{\pm1,\pm1}$ induces an isomorphism between $\cR_{cc}$ and $\cR_{aa}$, which as we already learned are conjugate to one another.\footnote{Similar statements can be made for the $\cR_{ac}$ and $\cR_{ca}$ rings, by using the spectral flow operators $\cU_{\pm1,\mp1}$.}
This leads to the following suggestive fact:
\begin{proposition}\label{prop:poincare}{\rm (Poincar\'e duality)} Let $h^{q,\tilde q}$ be the number of elements in $\cR_{cc}$ with charges $(q,\tilde q)$ with respect to the $N=(2,2)$ algebra. Then, these degeneracies satisfy the duality relation
  $$h^{q,\tilde q} =h^{\frac{c}{3}-q,\frac{\tilde c}{3}-\tilde q}.$$
\end{proposition}
\begin{proof}
\smartqed
First, recall that the weights of a $(c,c)$ operator must lie in the range
$$
(0,0)\leq(h,\tilde h)\leq (\tfrac c6,\tfrac{\tilde c}{6}),
$$
which means the charges obey
$$
(0,0)\leq (q,\tilde q) \leq (\tfrac{c}{3},\tfrac{\tilde c}{3}).
$$
Therefore the ``dual" charges, $(\frac{c}{3}-q,\frac{\tilde c}{3}-\tilde q)$, are well-defined for $(c,c)$ operators.

To demonstrate the duality, we will use the two equivalences between the $\cR_{cc}$ and $\cR_{aa}$ rings mentioned above. Charge conjugation takes the $\cR_{cc}$ ring to the $\cR_{aa}$ ring and reverses the signs of all the charges. Thus,
$$
h^{q,\tilde q}_{cc} = h^{-q,-\tilde q}_{aa},
$$
where the meaning of the subscripts should be obvious. At the same time, we can flow by $\th=\tilde\th=-1$ to return us to the $\cR_{cc}$ ring, which shifts all the charges by $(c/3,\tilde c/3)$, and so
$$
h^{-q,-\tilde q}_{aa} = h^{\tfrac{c}{3}-q,\tfrac{\tilde c}{3}-\tilde q}_{cc}.
$$
By a straightforward generalization, this result applies to the $\cR_{ac}$ ring as well.
\qed
\end{proof}
\begin{corollary}
  In every unitary $N=(2,2)$ theory, there exists a unique state in the $\cR_{cc}$ ring with the maximal weight $(h,\tilde h)=(c/6,\tilde c/6)$.
\end{corollary}
\begin{proof}
\smartqed
  We know that the vacuum is the unique state with $h=\tilde h =q =\tilde q =0$. Then by duality $h^{0,0} = h^{\tfrac{c}{3},\tfrac{\tilde c}{3}} =1$.
\qed
\end{proof}
This property of chiral rings brings us back to our earlier claim, that there is a deep connection between chiral rings of $N=(2,2)$ theories on the one hand, and the Dolbeault cohomology of complex manifolds on the other.
For such a geometric interpretation to make sense we should demand that $q,\tilde q\in\Z$ for all $(c,c)$ and $(a,c)$ states, and in particular this implies $c/3,{\tilde c/3}\in\Z$.\footnote{Alternatively, we could demand that the spectral flow operators $\cU_{\pm1,0},\cU_{0,\pm1}$ be {\it mutually local} with respect to the other operators in the theory, and the integrality condition will emerge automatically on the entire spectrum (not just the chiral sectors)~\cite{Banks:1987cy}.}
To obtain a geometric interpretation we should further impose $c=\tilde c=3n$, where now $n$ can be regarded as the (complex) dimension of the associated geometry.
In fact, once we insist on integral charges the chiral rings are naturally equipped with a grading and compatible Hodge structure:
$$
\cR_{cc} = \bigoplus_{q,\tilde q=0}^n \cR_{cc}^{q,\tilde q},\quad \cR_{ac} = \bigoplus_{q,\tilde q =0}^n \cR_{ac}^{-q,\tilde q},
$$
where we use $-q$ in the $(a,c)$ ring so that me can conveniently sum over $q\geq0$. This structure should be more or less obvious, so we will not belabour the details.
One point that does deserve our attention is the property of complex conjugation. In the $N=(2,2)$ theories this interchanges the holomorphic and anti-holomorphic sectors, which should be distinguished from the charge conjugation that acts within each sector by reversing the signs of all the charges.
This works straightforwardly for the $(c,c)$ ring:
$$
\overline{\cR_{cc}^{q,\tilde q}} = \cR_{cc}^{\tilde q,q},
$$
but for the $(a,c)$ ring we must be careful because this brings us to the $(c,a)$ ring. Of course these two rings are related by the second type of conjugation, and so we have
$$
\overline{\cR_{ac}^{-q,\tilde q}} = \cR_{ca}^{\tilde q,-q} = \cR_{ac}^{-\tilde q,q}.
$$

We now see that the chiral rings of $N=(2,2)$ theories, with integral $U(1)\times U(1)$ charges and central charge $c=\tilde c$, bear a striking resemblance to the cohomology rings of K\"ahler manifolds of dimension $c/3$.
In fact, under these assumptions there is an even stronger statement that can be made about the associated K\"ahler manifolds.
\begin{theorem}
Suppose that $q,\tilde q\in\Z$ for all states in the $(c,c)$ and $(a,c)$ sectors of a unitary $N=(2,2)$ superconformal theory with $c=\tilde c=3n$.
Then, the chiral rings $\cR_{cc}$ and $\cR_{ac}$ are formally equivalent to the Dolbeault cohomology rings of some {\rm Calabi-Yau} $n$-folds.
\end{theorem}
\begin{proof}
\smartqed
At this point, the relation between the ring structures has been laid out fairly explicitly. Let us briefly recapitulate the dictionary between the chiral rings and the cohomology rings:
\ban
\begin{array}{ccccc}
\underline{\ (c,c)\  } & \ &  \underline{\ Dolbeault\ } & \ & \underline{\ (a,c)\ }\\
G^+_{-1/2} & \longleftrightarrow & \del & \longleftrightarrow & G^-_{-1/2}\\
\tilde G^+_{-1/2} & \longleftrightarrow & \bar\del & \longleftrightarrow & \tilde G^+_{-1/2}\\
2L_0-J_0 & \longleftrightarrow & \Delta_{\del} & \longleftrightarrow &  2L_0 +J_0 \\
2\tilde L_0-\tilde J_0 & \longleftrightarrow & \Delta_{\bar\del} & \longleftrightarrow & 2\tilde L_0 -\tilde J_0 \\
J_0 & \longleftrightarrow & {\rm deg} & \longleftrightarrow & -J_0 \\
\tilde J_0 & \longleftrightarrow & \overline{{\rm deg}} & \longleftrightarrow & \tilde J_0 \\
\cR_{cc}^{p,q} & \longleftrightarrow & H^{p,q} & \longleftrightarrow & \cR_{ac}^{-p,q}
\end{array}
\ean
We have already seen that the chiral rings posses the properties we expect of a Hodge diamond, namely Poincar\'e duality and complex conjugation:
$$
h^{p,q} = h^{n-p,n-q} = h^{q,p},
$$
and also $h^{0,0} = h^{n,n}=1$ follows from uniqueness of the vacuum combined with spectral flow.

The novel claim of the theorem is that the associated K\"ahler manifolds should be Calabi-Yau, which additionally requires $h^{n,0}=h^{0,n}=1$.
To show this, consider the states
$$
\cU_{-1,0}\ket{0},\quad \cU_{+1,0}\ket{0},
$$
so that we only apply spectral flow to the holomorphic sector. These states are elements of $\cR_{cc}$ and $\cR_{ac}$, respectively, with charges $(n,0)$ and $(-n,0)$.
Similarly, by applying $\cU_{0,\mp1}$ to the vacuum we obtain the charge conjugate states. The uniqueness of the vacuum state then gives us
$$
h^{0,0} = h^{n,0} = h^{0,n} = h^{n,n} = 1,
$$
exactly as we expect for a (compact) Calabi-Yau $n$-fold.
\qed
\end{proof}
We should point out that, at this level of discussion, the relation between $(2,2)$ theories and Calabi-Yau manifolds is only a formal one.
That is, given the $\cR_{cc}$ and $\cR_{ac}$ rings of a $(2,2)$ theory that satisfy the conditions of the theorem, nothing guarantees the existence of a Calabi-Yau manifold with the appropriate Hodge numbers.
However, when no such manifold is known then the chiral rings of the $(2,2)$ theory in question are often used to {\it define} the space, at least at the topological level.

In addition, while we have demonstrated a (formal) equivalence between chiral primaries and cohomology classes, we have not shown an equivalence of ring structures.
That is, while we have two rings with the same degeneracies nothing guarantees that these rings have identical product structures.
In fact, we should not expect the two rings to be completely isomorphic, since we know that ``stringy" effects will modify the geometry when the space is sufficiently small.
We have already seen a prominent example of this in Sect.~\ref{ss:Tdual}\ when we discussed T-duality. In the large volume limit, when a conventional geometric description applies,
then the ring structures will indeed agree. We will not prove this fact here, but we will see this borne out in examples. However, when the size of the manifold is small then the chiral ring structure leads to significant modifications of the usual cohomology ring.
This structure is usually termed {\it quantum cohomology}, and this is the subject of Gromov-Witten theory. This, and related topics, are discussed in great detail in~\cite{Hori:2003ic}, among other places, and other notes appearing in this volume.

There is another important stringy effect encoded in the chiral rings, and this is {\it mirror symmetry}. While this relation appears mysterious and surprising from a geometric perspective,
it actually follows rather trivially as an automorphism of $(2,2)$ theories.
\begin{corollary}{\rm (Mirror symmetry)}
Given an $N=(2,2)$ SCFT subject to the conditions in the previous theorem, we can associate two (typically) distinct Calabi-Yau manifolds, $X$ and $Y$, related by
$$
H^{p,q}(X) \simeq H^{n-p,q}(Y).
$$
\end{corollary}
\begin{proof}
\smartqed
  If we identify $\cR_{ac}^{-p,q} = H^{p,q}(X)$ and $\cR_{cc}^{p,q} = H^{p,q}(Y)$, then by applying the spectral flow operator $\cU_{\mp1,0}$ we have
  $$
 H^{p,q}(X) = \cR_{ac}^{-p,q} \simeq \cR_{cc}^{n-p,q} = H^{n-p,q}(Y).
 $$
 \qed
\end{proof}
In other words, $(X,Y)$ form a mirror pair of Calabi-Yau $n$-folds.
In particular, the K\"ahler and complex structure deformations of the two spaces are interchanged: $H^{1,1}(X)\simeq H^{n-1,1}(Y)$ and vice-versa.
Recall that these geometric deformations correspond to the (truly) marginal operators $\hat\cO_{(\pm1,+1)}$ in the $N=(2,2)$ theory.
Which of these we choose to associate with K\"ahler deformations, and which we associate with complex structure deformations is therefore a matter of taste, because there is a fundamental ambiguity:
\ban
\cO_{(1,1)} \in \cR_{cc}^{1,1} \simeq H^{1,1}(Y) \simeq H^{n-1,1}(X),\\
\cO_{(-1,1)} \in \cR_{ac}^{-1,1} \simeq H^{1,1}(X) \simeq H^{n-1,1}(Y).
\ean

As a relation between cohomology groups, the statement that for every Calabi-Yau $X$ there exists a mirror $Y$ with $h^{p,q}(X)=h^{n-p,q}(Y)$ is somewhat surprising, but not terribly deep.
The real depth of this relation comes from the fact that this isomorphism extends to the full (quantum) cohomology rings.
Mathematically, this leads to powerful computational techniques in enumerative geometry, among other areas.
For example, the computation of Gromov-Witten invariants, which is related to the  vertical cohomology, $H_{\rm vert}(X) = \oplus_k H^{k,k}(X)$, maps to a computation of periods in the middle dimension (horizontal) cohomology of the mirror, $H_{\rm hor}(Y) = \oplus_k H^{n-k,k}(Y)$.
In particular, the former requires the full {\it quantum} cohomology ring, while the latter can be carried out in using the {\it classical} cohomology. This was one of the earliest applications of mirror symmetry~\cite{Candelas:1990rm}.
Physically, the implications are very far reaching. Essentially the chiral rings control all of the physical observables.
Since a given ring can be associated with two distinct geometries there is no way for a string to distinguish between them.
This is much like the we saw in the case of T-duality, except now $X$ and $Y$ are topologically distinct manifolds.
It cannot be emphasized enough how surprising this relation is geometrically, and serves to underscore how differently strings ``see" the world around them as compared to a point particle.

\subsection{Examples}\label{ss:examples2}
We will now present several examples of $N=(2,2)$ SCFTs to illustrate the concepts that we have developed. Some of these will be formulated explicitly in terms of Calabi-Yau target spaces, making a direct connection to the chiral ring structure.
Other theories will have no a priori geometric interpretation, but one will emerge, through the chiral rings, nonetheless.

\subsubsection*{The Torus: Redux}
For our first example, we revisit our starting point from Sect.~\ref{sec:torus}\ and consider the torus, albeit from a much more formal perspective.
This is really nothing more than a combination of the free field CFTs of Sect.~\ref{ss:examples}, in a manner compatible with $(2,2)$ superconformal invariance.
We begin with two scalars compactified on circles of unit radii: $X^i \simeq X^i +2\pi$, for $i=1,2$.  It will prove useful to combine these into the complex scalar
$$
Z(z,\bar z) =\frac{1}{\sqrt{2}}(X^1 + \I X^2).
$$
$(2,2)$ supersymmetry requires an equal number of fermionic degrees of freedom, so we include a complex fermion in both the holomorphic and anti-holomorphic sectors:
$$
\psi(z)=\frac{1}{\sqrt{2}}(\psi^1+\I \psi^2),\qquad  \tilde\psi(\bar z)=\frac{1}{\sqrt{2}}(\tilde\psi^1+\I \tilde\psi^2).
$$
Each real scalar contributes $c=\tilde c=1$, while each Majorana (\ie\ real) fermion contributes either $c=\hlf$ or $\tilde c=\hlf$, depending on its holomorphicity.
Altogether, this theory has central charge $c=\tilde c =3$, consistent with our expectation for a Calabi-Yau of dimension $n=c/3 =1$.

The action of this  theory is that of a combination of free fields:
$$
S = \frac{1}{2\pi}\int \D^2z\left(\del Z\bar\del Z^* + \psi^*\bar\del\psi + \tilde\psi^*\del\tilde\psi\right),
$$
which we were studied separately in Sect.~\ref{ss:examples}. In terms of complex fields, the relevant OPEs in the holomorphic sector are given by:
\ban
\del Z(z) \del Z^*(w) &\sim -\frac{1}{(z-w)^2},\qquad
\psi^*(z)\psi(w)  \sim \frac{1}{z-w}.
\ean
Using these OPEs, it is a straightforward exercise to verify that the operators
\ban
T(z) &= -\del Z\del Z^* -\hlf \psi^*\del\psi - \hlf\psi\del\psi^*, \\
G^+(z) &= \sqrt{2}\I \, \psi^*\del Z, \\
G^-(z) &= \sqrt{2}\I \, \psi\del Z^*, \\
J(z) &= \psi^*\psi,
\ean
realize the $N=2$ algebra with $c=3$. Notice that $\psi^*(z)$ is a chiral primary operator of charge $q=+1$, while $\psi(z)$ is an anti-chiral primary of charge $q=-1$.
The scalar fields $Z,Z^*$ are neutral. Analogous results hold in the anti-holomorphic sector.
In particular, the (anti-)chiral rings are generated by the fermion fields from both sectors, so we have:
\ban
\cR_{cc} &= \{1,\psi^*,\tilde\psi^*,\psi^*\tilde\psi^*\}\\
\cR_{ac} &= \{1,\psi,\tilde\psi^*,\psi\tilde\psi^*\}
\ean
Note that $\psi^*\tilde\psi^*$ has $(h,\tilde h) = (\hlf,\hlf) = (c/6,\tilde c/6)$, which is maximal, so we have exhausted the possible elements in $\cR_{cc}$, and similarly for $\cR_{ac}$.
If we identify fermionic fields with differential one-forms, which is rather natural since both objects anti-commute, then the relation to the cohomology of the torus is clear:
$$
\cR_{cc}\simeq \cR_{ac} \simeq H^{*,*}(T^2) = \{1,\D Z,\D \bar Z, \D Z \D\bar Z\}.
$$

Let us examine the relation between the deformations of the torus and marginal operators in the field theory.
For this purpose, it helps to write the action as
$$
S = \frac{1}{4\pi} \int \D^2z\left(g_{z\bar z}\del Z\bar\del Z^* + g_{\bar z z}\del Z^* \bar\del Z +\ldots\right),
$$
where our starting point is $g_{z\bar z}=g_{\bar z z} =1$.
In each ring there is a unique operator that gives rise to a (truly) marginal deformation. In $\cR_{ac}$, that operator is
$$
\cO_{(-1,1)}(z,\bar z) = \psi(z)\tilde\psi^*(\bar z),
$$
which we act upon with $G^+(w)\tilde G^-(\bar w)$, and take a double contour integral, to produce the marginal operator
$$
\hat\cO_{(-1,1)}(z,\bar z) \propto \del Z(z)\bar \del Z^*(\bar z).
$$
This deformation has the same form as the (bosonic part of the) original action, and so adding this term to the action is equivalent to rescaling the torus metric $g_{z\bar z}$.
Thus, we have identified the K\"ahler deformation, which changes the area of the torus.
By adding a {\it complex} multiple of $\hat\cO_{(-1,1)}$,
$$
S \rightarrow \frac{1}{4\pi}\int \D^2z\left((1+\la)\del Z\bar\del Z^* + (1+\bar\la)\del Z^*\bar\del Z +\ldots \right),
$$
we end up deforming the complexified K\"ahler class of the torus,
$$
B+\I \omega \in H^1(X,T^*X)\simeq H^{1,1}(X).
$$
This is exactly what we found earlier in Sect.~\ref{ss:Kahler}. In the case of $\cR_{cc}$, we have the chiral primary operator
$$
\cO_{(1,1)}(z,\bar z) = \psi^*(z)\tilde\psi^*(\bar z),
$$
which we act upon by $G^-(w)\tilde G^-(\bar w)$, and take a double contour integral, to produce the marginal deformation
$$
\hat\cO_{(1,1)} \propto \del Z^*(z) \bar\del Z^*(\bar z).
$$
This deforms the metric by adding the non-Hermitian components $\dd g_{\bar z \bar z}$ and its complex conjugate.
However, we can always restore the Hermitian structure of the metric by changing the complex structure:
$$
\D Z \rightarrow \D Z' = \D Z + g^{z\bar z}\dd g_{\bar z\bar z}\D\bar Z.
$$
Thus, the deformation $\hat\cO_{(1,1)}$ induces a change in the complex structure given by
$$
\chi^z_{\bar z} = g^{z\bar z}\dd_{\bar z\bar z} \in H^1(X,TX)\simeq H^{0,1}(X).
$$
As expected, the moduli of the torus precisely match the deformations of the $(2,2)$ SCFT, which in turn are controlled by the chiral rings of the theory.
This simple example illustrates these features very nicely, and more elaborate geometric models will only build upon this basic structure.

\subsubsection*{Nonlinear Sigma Models}
Let us extend the previous example of the torus to more general target manifolds, $X$. Remarkably, the existence of $(2,2)$ supersymmetry imposes the restriction that $X$ must be a K\"ahler manifold~\cite{Zumino:1979et,AlvarezGaume:1981hm}.\footnote{To be precise, $X$ is K\"ahler only when the $B$-field is closed. More generally, $X$ can be bi-Hermitian~\cite{Gates:1984nk}\ or equivalently Generalized K\"ahler~\cite{2004math......1221G}.}
Furthermore, conformal invariance requires that $X$ be Ricci-flat~\cite{PhysRevLett.45.1057}. Together, these restrictions single out Calabi-Yau geometries as viable $(2,2)$ target spaces.
Working in complex basis, $X$ can parameterized by local coordinates $Z^i$ and $Z^\ibar = Z^i{}^*$, and Hermitian metric $g_{i\jbar} = g_{i\jbar}(Z,Z^*).$
For $n$ such complex scalars, together with their fermionic partners $\psi^i,\tilde\psi^i$ and $\psi^\ibar,\tilde\psi^\ibar$, the conformal theory will have central charge $c=\tilde c = 3n$, consistent with our expectation of a Calabi-Yau $n$-fold.
Unlike the case of the (flat) torus, the action for a general target, known as the nonlinear sigma model, will not involve free fields, and is given by:
$$
S = \frac{1}{2\pi}\int \D^2z\left( \hlf g_{i\jbar} \left(\del Z^i\bar\del Z^{\jbar} + \del Z^{\jbar}\bar\del Z^i\right)
+ g_{i\jbar}\psi^{\jbar}\bar D \psi^i + g_{i\jbar}\tilde\psi^{\jbar} D\psi^i +R_{i\jbar k\bar\ell}\psi^i\psi^{\jbar} \tilde\psi^k\tilde\psi^{\bar\ell}\right),
$$
where the covariant derivatives acting on the fermions are given by
\ban
\bar D \psi^i &= \bar\del \psi^i + \G^i_{jk}\bar\del Z^j \psi^k, \qquad
D \tilde \psi^i = \del\tilde\psi^i + \G^i_{jk}\del Z^j \tilde\psi^k,
\ean
and $\G$ and $R$ are the Levi-Civita connection and curvature associated with the metric. Note that $g,\G$ and $R$ are all functions of $Z,Z^*$, and so the action is highly nonlinear.
Of course for a flat metric, the nonlinear sigma model reduces to the free field case of a (complex) $n$-torus.

The $N=2$ algebra, with $c=3n$, is realized by the following combinations of operators:
\ban
T(z) &= -g_{i\jbar}\,\del Z^i\del Z^{\jbar} -\hlf g_{i\jbar} \left(\psi^{\jbar}\del\psi^i + \psi^i\del\psi^{\jbar}\right) +\ldots, \\
G^+(z) &= \sqrt{2}\I  g_{i\jbar}\,\psi^{\jbar}\del Z^i, \\
G^-(z) &= \sqrt{2}\I  g_{i\jbar}\,\psi^i\del Z^{\jbar}, \\
J(z) &= g_{i\jbar}\,\psi^{\jbar}\psi^i,
\ean
where we have suppressed higher order terms in $T(z)$, which will not be needed for our discussion.
This can be verified by employing the OPEs:
\ban
\del Z^i(z) \del Z^{\jbar}(w) &\sim -\frac{g^{i\jbar}}{(z-w)^2}, \qquad \psi^i(z)\psi^{\jbar}(w) \sim \frac{g^{i\jbar}}{z-w}.
\ean
Once again, the ring $\cR_{ac}$ is generated by the fermion fields $\psi^i$ and $\tilde\psi^{\jbar}$. Consider the operator
$$
\cO_{(-p,q)} = \omega_{i_1\ldots i_p\jbar_1\ldots\jbar_q} \psi^{i_1}\ldots\psi^{i_p}\tilde\psi^{\jbar_1}\ldots\tilde\psi^{\jbar_q},
$$
with charges $(-p,q)$, where $\om_{(p,q)}= \om_{(p,q)}(Z, Z^*)$ is an arbitrary coefficient function (which does not affect the weights or charges of this operator).
As before, if we replace the fermions by one-forms, we can think of $\om_{(p,q)}$ as $(p,q)$-form on $X$.
Then, by taking OPEs with $G^+$ and $\tilde G^-$, it follows that $\cO_{(-p,q)}$ is in the $(a,c)$ ring if and only if $\del\om_{(p,q)}=\bar\del\om_{(p,q)}=0$, or in other words:
$$
\cO_{(-p,q)}\in\cR_{ac}^{-p,q}\quad \Leftrightarrow\quad \om_{(p,q)}\in H^q(X,\wedge^p T^*X) \simeq H^{p,q}(X).
$$
In particular, the marginal operators
$$
\hat\cO_{(-1,1)}(z,\bar z) = \om_{i,\jbar}\ \del Z^i(z)\bar\del Z^{\jbar}(\bar z) + (fermions)
$$
contain the K\"ahler metric deformations of $X$, together with the necessary modifications to the fermionic terms in the action.
Similarly, $\cR_{cc}$ is generated by the fermionic fields $\psi_i = g_{i\jbar}\psi^{\jbar}$ and $\tilde\psi^{\jbar}$. A generic operator of charge $(p,q)$ takes the form
$$
\cO_{(p,q)} = \chi^{i_1\ldots i_p}_{\jbar_1\ldots\jbar_q} \psi_{i_1}\ldots\psi_{i_p}\tilde\psi^{\jbar_1}\ldots\tilde\psi^{\jbar_q},
$$
where $\chi^{(p)}_{(q)} = \chi^{(p)}_{(q)}(Z,Z^*)$ is also an arbitrary coefficient function. Then, by the same reasoning as above,
$$
\cO_{(p,q)}\in\cR_{cc}^{p,q}\quad \Leftrightarrow\quad \chi^{(p)}_{(q)}\in H^q(X,\wedge^p TX) \simeq H^{n-p,q}(X).
$$
The bosonic terms in the marginal operator $\hat\cO_{(1,1)}$ deform the complex structure of the metric:
$$
\hat\cO_{(1,1)}(z,\bar z) = \chi_{\jbar}^i\ g_{i\bar k}\del Z^{\bar k}(z)\bar\del Z^{\jbar}(\bar z) + (fermions).
$$
 Of course, spectral flow/mirror symmetry will interchange the two classes of  marginal operators, and so associating one with K\"ahler as opposed to complex structure deformations is purely conventional.

Thus, when a $N=(2,2)$ superconformal theory can be realized geometrically by a nonlinear sigma model, there is a direct connection between the chiral rings of the theory and the cohomology of the (Calabi-Yau) target space.
This relation goes back to the pioneering work of Witten~\cite{Witten:1982df}. There, the Ramond groundstates of $N=(1,1)$ supersymmetric sigma models were first identified with the cohomology of the target spaces.
In $(2,2)$ theories, spectral flow extends this identification to the chiral rings in the NS sectors.

\subsubsection*{Landau-Ginzburg Orbifolds}
In the sigma model examples above a Calabi-Yau manifold is needed as an input to specify the theory, so it comes as no surprise when that Calabi-Yau's cohomology shows up in the theory's chiral rings.
However, from the general discussion of Sect.~\ref{ss:geometry}\ we expect this should happen for any $(2,2)$ SCFT with integral charges, even when there is no obvious geometric interpretation.
Here we would like to present one such class, and examine how the Calabi-Yau emerges from the $N=2$ structure

The theories we will consider are based on so-called {\it Landau-Ginzburg} models. These theories are most succinctly described in superspace, but developing that formalism here in detail would spoil all of its economy.
Interested readers could consult~\cite{Hori:2003ic}\ for a nice introduction to the subject. The basic objects we use are chiral superfields $\Phi^i$, which contain both scalar fields, $\phi^i$, and fermion fields $\psi^i,\tilde \psi^i$ as components.
Landau-Ginzburg models are completely characterized by a certain holomorphic function, called a {\it superpotential}, $W(\Phi)$.
In order for a Landau-Ginzburg model to admit a (super)conformal symmetry, it must be a quasi-homogeneous function:
$$
W(\la^{w_i}\Phi^i) = \la^d W(\Phi^i),
$$
for some integers $w_i,d$, which contain no common factors. In order that this theory be non-degenerate, which is analogous to being having a compact target space,
$W$ must have an isolated minimum at the origin:
$$
\del_i W(\Phi) = 0,~\forall~i\quad \Leftrightarrow \quad \Phi^i=0,~\forall~i.
$$

When formulating a suitable action for these models, the superpotential is integrated against a measure on superspace that carries charges $q=\tilde q = -1$, and weights $h=\tilde h=-\hlf$.
Therefore, superconformal invariance requires that $W(\Phi)$ transforms in the opposite manner, with $q_W=\tilde q_W=+1$ and $h_W=\tilde h_W =+\hlf$.
Given the quasi-homogeneity of $W$, this completely determines the charges and weights of the chiral superfields $\Phi^i$:
$$
q_i =\tilde q_i =\frac{w_i}{d},\qquad h_i=\tilde h_i = \frac{w_i}{2d}.
$$
In particular, the operators $\Phi^i$ are chiral primaries, and generate the chiral ring. However, not all possible combinations of $\Phi^i$ correspond to non-trivial chiral ring elements.
The reason is that any operator of the form $\del_i W(\Phi)$ is (by the classical equations of motion) $G^+_{-1/2}$-exact~\cite{Lerche:1989uy}, and therefore by Thm.~\ref{thm:coh}\ is trivial in the chiral ring.
Thus, the chiral ring of a Landau-Ginzburg model is given by the local ring of $W(\Phi)$:
$$
\cR^0_{cc} = \frac{\CC[\Phi]}{\del_i W(\Phi)},
$$
where the 0 superscript reminds us that this is not the ring we are ultimately interested in, but only a starting point.
It is easy to see that the (unique) state of highest weight, $h=\tilde h =c/6$, in this ring is given by
$$
\cU_{1,1} = \det \left[\del_i\del_j W(\Phi)\right],
$$
which must correspond to the spectral flow operator. A quick count reveals that this highest weight state has weight
$$
\frac{c}{6} = \sum_i\left(\hlf-q_i\right).
$$
Since $q_i=\tilde q_i$ for all $\Phi^i$, and similarly for all chiral primary operators as well, the $(a,c)$ rings must actually trivial in these models, since there are no operators with negative charges.
This does not bode well for uncovering a relation to Calabi-Yau manifolds or mirror symmetry in these models.
However, recall that in order to assign a geometric interpretation to the chiral ring we must have integral charges, which is clearly violated here.
Furthermore, the ``dimension"
$$
n=\frac{c}{3} = \sum_i (1-2q_i)
$$
must be integral, which is also generally false. Clearly, as they stand Landau-Ginzburg models are not quite suitable for our purposes.

Notice that all of the chiral superfields, $\Phi^i$, have charges that are multiples of $1/d$: $q_i=w_i/d$.
So while the Landau-Ginzburg model itself will not contain integral charges, a $\Z_d$-{\it orbifold} of it~\cite{Dixon:1985jw}\ will.
What this means is to take a $\Z_d$ quotient of the theory, projecting the full Hilbert space onto the sub-sector invariant under the discrete subgroup $\Z_d\subset U(1)$.
The only states in $\cR^0_{cc}$ that survive this projection are guaranteed to have integral charges.
However, this quotienting procedure introduces new states, so-called {\it twisted states}, into the theory which are only quasiperiodic under $z\rightarrow \E^{2\pi\I}z$ up to $d$-th roots of unity.
In particular, a state which transforms as
$$
\Phi^i(\E^{2\pi\I}z) = \E^{2\pi\I kq_i}\Phi^i(z),
$$
for $k=1,2,\ldots, d-1$, is said to be in the {\it $k$-th twisted sector}.
When $kq_i\in \Z$, then $\Z_d$ invariant combinations of $\Phi^i$ generate new states in these twisted sectors.
However, for $kq_i\in\!\!\!\!\!\!\!/\,\Z$, then $\Phi^i$ is fixed at the origin and the $k$-th twisted sector only contains its ground state.
This twisting induces a charge for the groundstate in each of these twisted sectors~\cite{Vafa:1989xc}:
\be\label{eqn:charges}
\left(\sum_i\left( [kq_i]-\frac12\right)+\frac c6 \ ,\sum_i\left(- [kq_i]+\frac12\right)+\frac c6\right),
\ee
where $[x]$ denotes the fractional part of $x$.\footnote{The shift by $c/6$ comes about by spectral flow from the Ramond sector.}
When $n=c/3\in\Z$, then these groundstate charges are guaranteed to be integral~\cite{Vafa:1989xc}, as required.
Notice that the groundstate in the $k=d-1$ twisted sector has charges $(c/3,0)$, which corresponds to the holomorphic spectral flow operator $\cU_{1,0}$.
Using (the inverse of) this operator we can generate an isomorphic $\cR_{ac}$ ring by spectral flow, which was certainly not the case before performing the orbifold.
Thus, {\it given a Landau-Ginzburg model with superpotential $W(\Phi)$ of quasi-homogeneous degree $d$ and central charge $c=3n$ (for $n\in\Z$), we can construct a $\Z_d$-orbifold of the theory such that all charges are integral}.
In particular, the chiral rings of such Landau-Ginzburg orbifolds should exhibit the familiar Calabi-Yau structure we have some to expect.

There are many, many possible models that can be constructed in this manner. Let us focus on a particularly simple class to illustrate the essential points.
Assume all of the chiral superfields have the same weighting, so that $w_i=1$. Then, one simple way to satisfy requirements is to take $d$ superfields, $\Phi^i$, with a superpotential
$$
W(\Phi) = \sum_{i=1}^d (\Phi^i)^d.
$$
So all superfields have charge $q=1/d$, and we should expect to uncover the structure of some Calabi-Yau manifold of dimension
$$
n=\frac{c}{3} = \sum_{i=1}^d \left(1-\frac2d\right) = d-2.
$$
The first non-trivial case is $d=3$, where we should hope to discover a torus (once again). Let us see how this comes about in detail.
We begin with $\CC[\Phi^1,\Phi^2,\Phi^3]$, the polynomial ring on $\CC^3$, and quotient by the ideal $\langle(\Phi^1)^2,(\Phi^2)^2,(\Phi^3)^2\rangle$.
This gives the chiral ring
$$
\cR^0_{cc} = \{1,\Phi^1,\Phi^2,\Phi^3,\Phi^1\Phi^2,\Phi^2\Phi^3,\Phi^3\Phi^1,\Phi^1\Phi^3\Phi^3\}
$$
in the un-orbifolded theory. Most of these operators have fractional charges. The only operators that survive the orbifold process are just $\{1,\Phi^1\Phi^2\Phi^3\}$ with charges $(0,0)$ and $(1,1)$, respectively.
Finally, we must include the twisted sectors. Since $q_i=1/3$, only the groundstates contribute in the $k=1,2$ twisted sectors.
Using the formula~\C{eqn:charges}, we see that the $k=1,2$ twisted sectors generate chiral ring elements with charges
$$
(0,1)\quad {\rm and} \quad (1,0),
$$
respectively. Since $c=3$, these are just standard the spectral flow operators $\cU_{0,1}$ and $\cU_{1,0}$.
Altogether, the $\Z_3$-orbifold of the Landau-Ginzburg model with $q_1=q_2=q_3=1/3$ has the chiral ring
$$
\cR_{cc} = \{1, \cU_{1,0},\cU_{0,1},\Phi^1\Phi^2\Phi^3\}.
$$
Comparing charges, we see  that this is isomorphic to what we found above for the torus sigma-model.
Although the Landau-Ginzburg model was formulated without any explicit reference to a target geometry, perhaps the emergence of a torus could have been anticipated.
If we regard the fields $\Phi^i$ as the homogeneous coordinates of a $\P^2$, then the hypersurface
$$
\left\{(\Phi^1)^3 + (\Phi^2)^3 + (\Phi^3)^3 + \la \Phi^1\Phi^2\Phi^3 =0\right\}\subset \P^2
$$
defines a one-parameter family of elliptic curves (\ie\ tori). At $\la=0$ we have our original superpotential $W$, and we see quite clearly that this family of tori are connected by the unique chiral primary with charges $(1,1)$, as expected.

The above construction generalizes readily to more complicated Calabi-Yau manifolds. For example, following the exact same procedure as above but with $d=4$ leads to
the quartic $K3$ surfaces in $\P^3$, while taking $d=5$ generates the quintic three-folds in $\P^4$. By assigning non-uniform weights, $w_i$, to the fields it is possible to generate Calabi-Yau hypersurfaces in {\it weighted} projective space, $\P^{d-1}_{w_1,\ldots,w_d}$.
Allowing fields to carry multiple weights leads to hypersurfaces in products of (weighted) projective spaces, while including multiple superpotentials extends these constructions to complete intersection Calabi-Yaus in general toric varieties.

For years this so-called {\it Landau-Ginzburg/Calabi-Yau correspondence} was rather mysterious. In~\cite{Witten:1993yc}, Witten demonstrated that the Landau-Ginzburg and nonlinear sigma models are really just different limits of the same underlying SCFT.
His approach was to construct a two-dimensional gauge theory that flows to either limit depending on how some parameters are chosen.
In a sense, the Landau-Ginzburg models emerge in the limit of the sigma model when the volume of the Calabi-Yau manifold (formally) becomes negative.
Of course, we should not trust the geometric picture associated with the sigma model once the volume becomes sufficiently small.
A more precise statement would be that the moduli of the conformal field theory, which we associate with K\"ahler deformations at large volumes, become negative as we approach the Landau-Ginzburg point in the moduli space.
Once again, we see that the conformal field theory that lives on the string's worldsheet perceives geometry very differently from the way we might expect, and it able to make sense of seemingly singular spaces.

\subsection{Applications}\label{ss:summary2}
 Having spent all our time explaining how mirror symmetry arises in physics, we have no time left to explain its wealth of applications in detail.
Let us briefly comment on just two of them. Perhaps most famously, the authors of~\cite{Candelas:1990rm}\ used mirror symmetry to predict the number of rational curves of fixed degree in the quintic three-fold.
This problem in enumerative geometry, which depends of the manifold's K\"ahler structure, arises physically as a computation of instanton corrections, and is notoriously difficult.
However, on the mirror manifold this maps to a question involving the variation of Hodge structure, which turns out to be rather straightforward.

Another remarkable application of mirror symmetry is the demonstration that a change in the topology of spacetime, usually considered a rather violet procedure, can proceed smoothly in string theory.
The idea used~\cite{Aspinwall:1993nu}\ (see also~\cite{Witten:1993yc}) is to follow a topological transition from the point of view of the mirror manifold.
In the transition they study, the so-called {\it flop}, a two-cycle in the original Calabi-Yau shrinks to zero size, thereby producing a singular manifold which is then resolved (\ie\ {\it blown-up}) to produce a new, topologically distinct, (smooth) Calabi-Yau.
Such operations abound in the study of birational geometry. What~\cite{Aspinwall:1993nu}\ found is that the mirror of the flop is an innocuous change in the complex structure of the mirror Calabi-Yau.
In particular, the CFT on the mirror side is perfectly well-behaved, and so we must conclude that the same holds for the flop, despite the intermediary singular manifold.
Once again, we are amazed at how differently strings view spacetime as compared to point particles.

Now, at the very end of these notes, we have really only reached the starting point. We have developed the basic tools and language of (super-)conformal field theories, so that we could understand how mirror symmetry arises in that context:
namely, as an ambiguity in assigning a Calabi-Yau geometry to an $N=(2,2)$ SCFT.
We hope to have provided the reader with sufficient background that they may freely study the physics literature, to better understand some of the more recent applications and advances in this fascinating topic.



\end{document}